\newif\ifcolourslides \colourslidestrue
                         \colourslidesfalse
\documentclass{llncs}

\usepackage{amsmath}
\usepackage{url}

\usepackage{amsthm}

\usepackage{mathtools,amssymb,proof,stmaryrd}

\usepackage[noline,linesnumbered]{algorithm2e}
\SetInd{0.25em}{0.25em}
\SetAlFnt{\scriptsize\sf}

\usepackage{listings}
\usepackage{prooftree}
\usepackage{enumitem}

\usepackage{cleveref}
\crefname{theorem}{thm.}{theorems}
\crefname{definition}{defn.}{definitions}
\crefname{proposition}{prop.}{propositions}
\crefname{figure}{fig.}{figures}
\crefname{enumi}{condition}{conditions}
\crefname{section}{sec.}{sections}

 \theoremstyle{plain}
 \newcommand{\thistheoremname}{}
 \newtheorem*{genericthm}{\thistheoremname}
 \newenvironment{ourproblem}[1]
   {\renewcommand{\thistheoremname}{#1}\begin{genericthm}}{\end{genericthm}}

\makeatletter
\newcommand{\dotminus}{\mathbin{\text{\@dotminus}}}

\newcommand{\@dotminus}{%
  \ooalign{\hidewidth\raise1ex\hbox{.}\hidewidth\cr$\m@th-$\cr}%
}
\makeatother

\DeclarePairedDelimiter{\set}{\lbrace}{\rbrace}
\renewcommand{\vec}[1]{\mathbf{#1}}

\newcommand{\defeq}{=_{\textrm{\scriptsize{def}}}}

\newcommand{\nat}{\mathbb{N}}

\newcommand{\partialfn}{\rightharpoonup}

\newcommand{\finpartialfn}{\partialfn_{\textrm{\tiny fin}}}

\newcommand{\var}{\mathsf{Var}}
\newcommand{\nil}{\mathsf{nil}}

\newcommand{\emp}{\mathsf{emp}}
\newcommand{\psto}[2]{{#1} \mapsto {#2}}

\newcommand{\ASL}{\mathsf{ASL}}
\newcommand{\SL}{\mathsf{SL}}

\newcommand{\pres}{\mathsf{PbA}}

\newcommand{\dom}[1]{\mathrm{dom}\left(#1\right)}

\newcommand{\nilv}{\mathit{nil}}

\newcommand{\bigsepstar}{\mathop{\raisebox{-1ex}{{\Huge $*$}}}}

\newcommand{\ls}[2]{\mathsf{ls}\,#1\,#2}

\newcommand{\authorcomment}[2]{
\begin{center}
\fbox{
\begin{minipage}{.9\textwidth}
{\bf #1:} {\it #2}
\end{minipage}}
\end{center}}

\long\def\maxcomment#1{\authorcomment{Max' comment}{#1}}

\newcommand{\NP}{\ensuremath{\mathsf{NP}}}
\newcommand{\CoNP}{\ensuremath{\mathsf{coNP}}}
\newcommand{\EXPTIME}{\ensuremath{\mathsf{EXP}}}
\newcommand{\EXP}{\EXPTIME}

\long\def\comment#1{}
\long\def\commentt#1{}
\comment{

} 

\renewcommand{\paragraph}[1]{\vspace{2pt}\noindent\textit{#1}}

\long\def\comment#1{}
\ifcolourslides \usepackage{color}\input{rgb} 
\else \def\color#1{}
\fi

\newcount \WIDTH  
\newcount \HEIGHT 
\newcount \Xcur   
\newcount \Ycur   
\newcount \HALF   
\newcount \DBL    
\newcount \QUA    
\newcount \Qz     
\newcount \Qx     
\newcount \Qxx     
\newcount \TreeH
\newcount \TreeW

\long\def\comment#1{}

\def\blue#1{\mbox{#1}}

\def\red#1{\mbox{#1}}
\def\Dred#1{\mbox{#1}}

\def\sep{\mid} 

\def\CircNode#1{\circle{16}\makebox(0,0)[c]{#1}}

\def\node#1{\widehat{#1}}

\def\WEarrow#1#2#3
{\begin{picture}(0,0)\thicklines%
 \QUA=#2\divide\QUA by 2%
 \DBL=#1\advance\DBL by -9%
 \HALF=#1\advance\HALF by -#2\advance\HALF by 11
 $\bezier{200}(-\DBL,5)(-\HALF,\QUA)(0,\QUA)$%
  \advance\DBL by -6\advance\HALF by -7
 $\bezier{200}(0,\QUA)(\HALF,\QUA)(\DBL,8)$%
 \put(\DBL,8){\vector(2,-1){6}}
 \put(0,\QUA){\makebox(0,0)[b]{\raisebox{3pt}{#3}}}
\end{picture}}

\def\WEArrow#1#2#3
{\begin{picture}(0,0)\thicklines%
 \QUA=#2\divide\QUA by 2%
 \DBL=#1
 \HALF=#1\advance\HALF by -#2\advance\HALF by 9
 $\bezier{200}(0,-\QUA)(-\DBL,-\QUA)(-\DBL,-8)$%
 \Qz=\DBL\advance\Qz by 4%
 $\bezier{200}(\DBL,-8)(\Qz,-\QUA)(0,-\QUA)$%
 \put(\DBL,-12){\vector(0,1){4}}
 \put(0,-\QUA){\makebox(0,0)[b]{\raisebox{2pt}{#3}}}
\end{picture}}

\def\EWArrow#1#2#3
{\begin{picture}(0,0)\thicklines%
 \QUA=#2\divide\QUA by 2%
 \DBL=#1
 \HALF=#1\advance\HALF by -#2\advance\HALF by 9
 \Qz=\DBL\advance\Qz by 4%
 $\bezier{200}(0,-\QUA)(-\Qz,-\QUA)(-\DBL,-8)$%
 $\bezier{200}(\DBL,-8)(\DBL,-\QUA)(0,-\QUA)$%
 \put(-\DBL,-12){\vector(0,1){4}}
 \put(0,-\QUA){\makebox(0,0)[b]{\raisebox{2pt}{#3}}}
\end{picture}}

\def\EWarrow#1#2#3
{\begin{picture}(0,0)\thicklines%
 \QUA=#2\divide\QUA by 2%
 \DBL=#1\advance\DBL by -9%
 \HALF=#1\advance\HALF by -#2\advance\HALF by 9
 $\bezier{200}(\DBL,-5)(\HALF,-\QUA)(0,-\QUA)$%
  \advance\DBL by -6\advance\HALF by -5
 $\bezier{200}(0,-\QUA)(-\HALF,-\QUA)(-\DBL,-8)$%
 \put(-\DBL,-8){\vector(-2,1){6}}
 \put(0,-\QUA){\makebox(0,0)[t]{\raisebox{-6pt}{#3}}}
\end{picture}}

\def\cNorthEast#1#2#3
{\begin{picture}(0,0)\thicklines%
 \HALF=#1 \advance \HALF by -15 
 \put(8,8){\line(1,1){\HALF}}
 \HALF=#1 \divide\HALF by 2
 \put(\HALF,\HALF){\makebox(0,0)[b#3]{\underline{#2}}}
 \HALF=#1 \advance \HALF by -22 
 \put(\HALF,\HALF){\vector(1,1){14}}
\end{picture}}

\def\cNorthWest#1#2#3
{\begin{picture}(0,0)\thicklines%
 \HALF=#1 \advance \HALF by -15 
 \put(-8,8){\line(-1,1){\HALF}}
 \HALF=#1 \divide\HALF by 2
 \put(-\HALF,\HALF){\makebox(0,0)[b#3]{\underline{#2}}}
 \HALF=#1 \advance \HALF by -22 
 \put(-\HALF,\HALF){\vector(-1,1){14}}
\end{picture}}

\def\cSouth#1#2#3
{\begin{picture}(0,0)\thicklines%
 \HALF=#1 \advance \HALF by -20 
 \put(0,10){\line(0,1){\HALF}}
 \HALF=#1 \divide\HALF by 2
 \put(0,\HALF){\makebox(0,0)[b#3]{\underline{#2}}}
 \put(0,24){\vector(0,-1){14}}
\end{picture}}

\def\cNorth#1#2#3
{\begin{picture}(0,0)\thicklines%
 \HALF=#1 \advance \HALF by -20 
 \put(0,10){\vector(0,1){\HALF}}
 \HALF=#1 \divide\HALF by 2
 \put(0,\HALF){\makebox(0,0)[b#3]{\underline{#2}}}
\end{picture}}

\long\def\smallModelXY#1
{\setlength{\unitlength}{1pt}\HEIGHT=#1
\advance\HEIGHT by 14
\begin{picture}(0,\HEIGHT)%
  \put(0,\HEIGHT){%
\begin{picture}(0,0)\thicklines
\Ycur=0%
 \put(-#1,-\Ycur){\CircNode{$\node{x}$}}%
 \put(#1,-\Ycur){\CircNode{$\node{y}$}}%
 \put(0,-\Ycur){\WEarrow{#1}{#1}{$0$}}
\Qz=#1\multiply\Qz by 2\advance\Qz by 12%
 \put(0,-\Ycur){\WEArrow{#1}{\Qz}{$-1$}}%
\end{picture}%
}\end{picture}} 

\long\def\smallModelYX#1
{\setlength{\unitlength}{1pt}\HEIGHT=#1
\advance\HEIGHT by 14
\begin{picture}(0,\HEIGHT)%
  \put(0,\HEIGHT){%
\begin{picture}(0,0)\thicklines
\Ycur=0%
 \put(-#1,-\Ycur){\CircNode{$\node{x}$}}%
 \put(#1,-\Ycur){\CircNode{$\node{y}$}}%
 \put(0,-\Ycur){\WEarrow{#1}{#1}{$0$}}
\Qz=#1\multiply\Qz by 2\advance\Qz by 12%
 \put(0,-\Ycur){\EWArrow{#1}{\Qz}{$-1$}}%
\end{picture}%
}\end{picture}} 

\def\blue#1{{\color{blue}{{\boldmath\bf\mbox{#1}}}}}

\def\red#1{{\color{red}{{\boldmath\bf\mbox{#1}}}}}

\def\Dred#1{{\color{DarkRed}{{\boldmath\bf\mbox{#1}}}}}

\def\sep{\ \Dred{{\boldmath $\mid$}}\ }

\newcommand{\wand}{%
  \mathrel{\mbox{$\hspace*{-0.03em}\mathord{-}\hspace*{-0.66em}
  \mathord{-}\hspace*{-0.36em}\mathord{*}$\hspace*{-0.005em}}}}

\begin{document}
\title{On the Complexity of Pointer Arithmetic in Separation Logic
\\ (an extended version) }

\author{James Brotherston\inst{1}
\and Max Kanovich\inst{1,2}}

\institute{University College London, UK
\and National Research University Higher School of Economics,
Russian Federation}

\maketitle

\begin{abstract}
We investigate the complexity consequences of adding pointer arithmetic
to separation logic.  Specifically, we study extensions
 of the points-to
fragment of symbolic-heap separation logic with various forms of
Presburger arithmetic constraints.

Most significantly, we find that, even in the minimal case
 when we allow only conjunctions
 of simple ``difference constraints'' $x' \leq x \pm k$
 (where $k$ is an integer), polynomial-time decidability
 is already impossible:
 satisfiability becomes $\NP$-complete,
 while quantifier-free entailment becomes $\CoNP$-complete
 and quantified entailment becomes $\Pi^P_2$-complete
($\Pi^P_2$ is the second class in the polynomial-time hierarchy)

 In fact we prove that the upper bound is the same, $\Pi^P_2$,
 even for the full pointer arithmetic but with a fixed pointer offset,
 where we allow any Boolean combinations
 of the elementary formulas\ \mbox{$(x'=x+k_0)$},\
 \mbox{$(x'\leq x+k_0)$},\ and\ \mbox{$(x'<x+k_0)$},
 and, in addition to the points-to formulas,
 we allow spatial formulas of the arrays the length
 of which is $\leq k_0$
 and lists which length is $\leq k_0$, etc,
  where $k_0$ is a fixed integer.

However, if we allow a significantly more expressive form of pointer
arithmetic --- namely arbitrary Boolean combinations of elementary
formulas over arbitrary pointer sums --- then
 the complexity increase is relatively modest
 for satisfiability and quantifier-free entailment: they are
 still $\NP$-complete and $\CoNP$-complete respectively,
 and the complexity appears to increase drastically
 for quantified entailments,
 which becomes $\Pi^{\mathrm{EXP}}_1$-complete.
\end{abstract}

\keywords Separation logic, pointer arithmetic, complexity.

\long\def\comment#1{}

\section{Introduction}
\label{sec:introduction}

\emph{Separation logic} ($\SL$)~\cite{Reynolds:02} is a well-known and popular Hoare-style framework for verifying the memory safety of heap-manipulating programs.  Its power stems from the use of \emph{separating conjunction} in its assertion language, where $A * B$ denotes a portion of memory that can be split into two disjoint fragments satisfying $A$ and $B$ respectively.  Using separating conjunction, the  \emph{frame rule} becomes sound~\cite{Yang-OHearn:02}, capturing the fact that any valid Hoare triple can be extended with the same separate memory in its pre- and postconditions and remain valid, which empowers the framework to scale to large programs (see e.g.~\cite{Yang-etal:08}).  Indeed, separation logic now forms the basis for verification tools used in industrial practice, notably Facebook's \textsc{Infer}~\cite{Calcagno-etal:15} and Microsoft's \textsc{SLAyer}~\cite{Berdine-Cook-Ishtiaq:11}.

Most separation logic analyses and tools restrict the form of assertions
to a simple propositional structure known as \emph{symbolic
  heaps}~\cite{Berdine-Calcagno-OHearn:04}.  Symbolic heaps are
(possibly existentially quantified) pairs of so-called ``pure'' and
``spatial'' assertions, where pure assertions mention only equalities
and disequalities between variables and spatial formulas are
$*$-conjoined lists of pointer formulas $x \mapsto y$ and data structure
formulas typically describing segments of \emph{linked lists}
($\ls{x}{y}$) or sometimes binary trees.  This fragment of the logic
enjoys decidability in polynomial time~\cite{Cook-etal:11} and is
therefore highly suitable for use in large-scale analysers.  However, in
recent years, various authors have investigated the computational
complexity of (and/or developed prototype analysers for) many other
fragments employing various different assertion constructs, including
user-defined inductive
predicates~\cite{Iosif-etal:13,Brotherston-etal:14,Brotherston-etal:16,Antonopoulos-etal:14,Chen-Song-Wu:17},
pointers with \emph{fractional
  permissions}~\cite{Le-Gherghina-Hobor:12,Demri-Lozes-Lugiez:17},
arrays~\cite{Brotherston-Gorogiannis-Kanovich:17,Kimura-Tatsuta:17},
separating \emph{implication}
($\wand$)~\cite{Calcagno-Yang-OHearn:01,Brochenin-Demri-Lozes:12},
reachability predicates~\cite{Demri-Lozes-Mansutti:18} and
arithmetic~\cite{Le-Sun-Chin:16,Le-Tatsuta-Sun-Chin:17}.

It is with this last feature, arithmetic, with which we are concerned in
this paper.  In general, assertions involving arithmetic arise naturally
and for obvious reasons when analysing arithmetical programs; moreover,
the use of \emph{pointer arithmetic}, where pointers are treated
explicitly as numerical addresses which can be manipulated
arithmetically, is a standard feature e.g. of C code.  We therefore set
out by asking the following question: \emph{How much pointer arithmetic can one
  add to separation logic and remain within polynomial time?}

Unfortunately, and perhaps surprisingly, the answer turns out to be: essentially none at all.

We study the complexity of symbolic-heap separation logic with pointers,
but no other data structures, when pure formulas are extended by
arithmetical constraints, in two variants. The first variant
encapsulates a minimal language for pointer arithmetic, allowing only
conjunctions
of ``difference constraints'' \mbox{$x \leq y \pm k$}
 (where $k$ is an integer), whereas the second is more expressive,
 allowing
 arbitrary Boolean combinations of elementary formulas over  arbitrary
 pointer-and-offset sums.

We certainly do \emph{not} claim that either fragment is appropriate for
practical program verification; clearly, lacking constructs for lists or
other data structures, they will be insufficiently expressive for most
purposes (although they might be practical e.g. for some concurrent
programs that deal only with shared memory buffers of a small fixed
size).  The point is that any practical fragment of separation logic
employing arithmetic will almost inevitably include our minimal language
and thus inherit its computational lower bounds.

Our complexity results for SL pointer arithmetic are summarised in
Table~\ref{t-summary}. Perhaps our most striking result is that, even
for the case of our minimal SL pointer arithmetic where only constant
pointer offsets and conjunctions are permitted, the satisfiability
problem is already $\NP$-complete.  On the other hand, the problem is
still in $\NP$ when we extend to full pointer arithmetic.  However,
there is at least one material difference between the two fragments:
minimal pointer arithmetic enjoys the \emph{small model property},
meaning that any satisfiable symbolic heap $A$ has a model of size
polynomial in the size of $A$, whereas this property fails for full
pointer arithmetic.

In the case of the entailment problem, the story is somewhat similar:
for quantifier-free entailments the problem becomes $\CoNP$-complete,
irrespective of whether we consider minimal or full pointer arithmetic.
However, the complexity appears to increase drastically for quantified
entailments,
where the problem is $\Pi^P_2$-complete for minimal pointer arithmetic
but $\Pi^{\mathrm{EXP}}_1$-complete for full pointer arithmetic.
($\Pi^P_2$~is the second class in
 the \emph{polynomial-time hierarchy}~\cite{Stockmeyer:77}
 and $\Pi^{\mathrm{EXP}}_1$~is the
 first class in the \emph{exponential-time hierarchy},
 which corresponds to $\Pi^0_2$ Presburger arithmetic~\cite{Haase:14}).

\begin{table*}[t]  
\begin{center}%
\begin{tabular}{@{}l|c|c}%
    & minimal pointer arithmetic   & full pointer arithmetic
\\\hline\hline%
 Satisfiability   & \NP-complete 
  & \NP-complete
\\\hline%
 Small model property & yes  & no
\\\hline%
 Entailment, quantifier-free & \CoNP-complete  
  & \CoNP-complete
\\\hline%
 Entailment, quantified & $\Pi^P_2$-complete   
  &
  $\Pi^{\mathrm{EXP}}_1$-complete. 
\\\hline%
\end{tabular}
\end{center}
\caption{Summary of complexity results.}
\label{t-summary}
\end{table*}

The remainder of this paper is structured as follows.  In
Section~\ref{sec:language} we define symbolic-heap separation logic with
pointer arithmetic, in both ``minimal'' and ``full'' flavours.
Sections~\ref{sec:satisfiability} and~\ref{sec:entailment} study the
satisfiability and entailment problems, respectively, for our minimal
and full versions of SL pointer arithmetic, establishing upper and lower
complexity bounds for all cases.
In Section~\ref{s-P2-upper} we establish the small model property
and thereby the $\Pi^P_2$ upper bound
for the quantified entailments within minimal pointer arithmetic.
 Section~\ref{sec:conclusion} concludes.

\section{Separation logic with pointer arithmetic}
\label{sec:language}
Here, we introduce our language of \emph{separation logic with pointer arithmetic}, building on the well-known ``symbolic heap'' fragment over pointers~\cite{Berdine-Calcagno-OHearn:04}.

Because we have to take into account the balance between the
arithmetical part and the spatial part of the language, we consider two
varieties of pointer arithmetic: a ``minimal'' fragment containing only
the bare essentials, and a ``full'' fragment allowing greater
expressivity. To show lower complexity bounds, we have to challenge the
fact that $\Sigma^0_1$ Presburger arithmetic is already
\mbox{$\NP$-}hard by itself; thus, to reveal the true memory-related
nature of the problem, we restrict the arithmetical part of the language
by restricting the pure part of our language to something so simple that
it can be processed in polynomial time.. This leads us to consider {\em
  minimal pointer arithmetic}, in which we allow only conjunctions of
`difference constraints' of the form \mbox{$x'=x\pm k$}, and
\mbox{$x'\leq x\pm k$} where $x$ and $x'$ are variables and $k$ is an
integer (even negation \mbox{$x'\neq x$} is not permitted).  On the
other hand, for \emph{upper} complexity bounds, it stands to reason that
we should aim for as much expressivity as possible while remaining
within a particular complexity class.  Thus we also consider {\em full
  pointer arithmetic}, in which arbitrary Boolean combinations of
elementary formulas over  arbitrary pointer sums are permitted.

\begin{definition}[$\SL$ pointer arithmetic]\label{d-full}
A \blue{symbolic heap} is given by
\begin{equation}%
  \exists \vec{z}.\ \Pi\colon F       \label{eq-heap}
\end{equation}%
where $\vec{z}$ is a tuple of variables from an infinite set $\var$, and $\Pi$ and $F$ are respectively \emph{pure} and \emph{spatial} formulas, defined below.

For \emph{full pointer arithmetic}, we define terms~$t$, pure formulas~$\Pi$,
 and spatial formulas~$F$ by the following grammar:
\begin{align*}
     t & \Coloneqq\, x \in \var \sep t+k  \sep t+t
\\ \Pi & \Coloneqq\, t = t 
    \sep t \leq t
    \sep t < t \sep \Pi \land \Pi \sep \Pi \lor \Pi  \sep \neg\Pi
 \\
 F & \Coloneqq\, \emp \sep t\mapsto t \sep t\mapsto \nil{}
 \sep F * F 
  \label{eq-full}
\end{align*}
where $k$ ranges over $\nat$.

For \emph{minimal pointer arithmetic}, we instead define terms~$t$, pure formulas~$\Pi$,
 and spatial formulas~$F$ by the following simpler grammar:
\begin{align*}
     t & \Coloneqq\, x \in \var \sep t+k 
\\ \Pi & \Coloneqq\, t = t 
    \sep t \leq t
    \sep t < t \sep \Pi \land \Pi 
 \\
 F & \Coloneqq\, \emp 
 \sep t\mapsto \nil{}
 \sep F * F 
\end{align*}
Whenever one of $\Pi, F$ is empty in a symbolic heap $\exists \vec{z}.\ \Pi\colon F$, we omit the colon.
\end{definition}

In the case of minimal $\SL$ pointer arithmetic, the pure part of a
symbolic heap is a conjunction of {\em `difference constraints'} of the
form \mbox{$x'=x\pm k$} or \mbox{$x'\leq x\pm k$},
 where $x$ and $x'$ are variables, and $k$ is a fixed offset in $\nat$. 
 The satisfiability of such formulas can be decided in
 polynomial time; see~\cite{Cormen-etal:09}.
 The crucial observation is: 
\begin{proposition}\label{p-Cormen}
A `circular' system of difference constraints\ $x_1\leq x_2 + k_{12}$,\
   \mbox{$x_2\leq x_3 + k_{23}$},\
 \dots,\ \mbox{$x_{m-1}\leq x_{m} + k_{m-1,m}$},
 \mbox{$x_{m}\leq x_{1} + k_{m,m+1}$} allows one to conclude
 that\ \mbox{$x_1-x_1 \leq \sum_{i=1}^m k_{i,i+1}$},\
 which is a contradiction iff the latter sum is negative.
\end{proposition}
Thus, considering our symbolic heaps in minimal pointer arithmetic
readdresses the challenge of establishing relevant lower bounds to the
spatial part of the language.

\long\def\satOne#1#2#3#4#5#6%
{\setlength{\unitlength}{1pt}\HEIGHT=#1\multiply\HEIGHT by 2%
\advance\HEIGHT by 21
\begin{picture}(0,\HEIGHT)%
  \put(0,\HEIGHT){%
\begin{picture}(0,0)\thicklines
\Ycur=-7%
 \put(0,-\Ycur){\CircNode{#3}}%
\advance \Ycur by #1
\advance \Ycur by #1%
\Xcur=#1\multiply\Xcur by 2%
 \put(0,-\Ycur){\CircNode{#2}}%
 \put(-\Xcur,-\Ycur){\CircNode{#4}}%
 \put( \Xcur,-\Ycur){\CircNode{#5}}%
\Xcur=#1%
 \put(-\Xcur,-\Ycur){\WEarrow{#1}{#1}{$\ -1$}}
 \put(-\Xcur,-\Ycur){\EWarrow{#1}{#1}{#6}}
 \put( \Xcur,-\Ycur){\WEarrow{#1}{#1}{#6}}
 \put( \Xcur,-\Ycur){\EWarrow{#1}{#1}{$-1$}}
\Xcur=#1\multiply\Xcur by 2%
 \put( -\Xcur,-\Ycur){\cNorthEast{\Xcur}{$-1$}{r}}
 \put( \Xcur,-\Ycur){\cNorthWest{\Xcur}{$-1$}{l}}

\Xcur=#1\multiply\Xcur by 2%
 \put(2,-\Ycur){\cSouth{\Xcur}{$\,-1$}{l}}
 \put(-3,-\Ycur){\cNorth{\Xcur}{#6}{r}}
\Qz=#1\multiply\Qz by 2\advance\Qz by 12%
 \put(0,-\Ycur){\WEArrow{\Xcur}{\Qz}{$-1$}}
\end{picture}%
}\end{picture}} 

\medskip\noindent{\bf Semantics.}
As usual, we interpret symbolic heaps in a stack-and-heap model; for convenience we consider both locations to be natural numbers, and values to be either natural numbers or the non-addressable null value $\nilv$. Thus a \emph{stack} is a function $s\colon\var \rightarrow \nat \cup \{\nilv\}$. We extend stacks over terms as usual:
$s(n) = n$, $s(\nil)=\nilv$ and $s(t_1 + t_2) = s(t_1) + s(t_2)$. 
If $s$ is a stack, $z \in \var$ and $v$ is a value, we write $s[z \mapsto c]$ for the stack defined as $s$ except that $s[z \mapsto v](z) = v$. We extend stacks pointwise over term tuples.

A \emph{heap} is a finite partial function $h\colon \nat \finpartialfn \nat$ mapping finitely many locations to values;
we write $\dom{h}$ for the domain of $h$, and $e$ for the empty heap that is undefined on all locations.
We write $\circ$ for \emph{composition} of domain-disjoint heaps: if $h_1$ and $h_2$ are heaps, then $h_1 \circ h_2$ is the union of $h_1$ and $h_2$ when $\dom{h_1}$ and $\dom{h_2}$ are disjoint, and undefined otherwise.

\begin{definition}
\label{defn:satisfaction}
The \emph{satisfaction relation} $s,h\models A$, where $s$ is a stack, $h$ a heap and $A$ a symbolic heap, is defined by structural induction on $A$. 
\[\begin{array}{l@{\hspace{0.1cm}}c@{\hspace{0.2cm}}l}
s,h \models t_1 \sim t_2 & \Leftrightarrow & s(t_1) \sim s(t_2) \text{ where } \sim \text{ is $=$,$<$ or $\leq$} \\
s,h \models \neg \Pi & \Leftrightarrow & s,h \not\models \Pi \\
s,h \models \Pi_1 \land \Pi_2 & \Leftrightarrow & s,h \models \Pi_1 \text{ and } s,h\models \Pi_2 \\
s,h \models \Pi_1 \vee \Pi_2 & \Leftrightarrow & s,h \models \Pi_1 \text{ or } s,h\models \Pi_2 \\
s,h \models \emp & \Leftrightarrow & h = e \\
s,h \models \psto{t_1}{t_2} & \Leftrightarrow & \dom{h}= \set{s(t_1)}\mbox{ and } h(s(t_1)) = s(t_2) \\
s,h \models F_1 * F_2 & \Leftrightarrow &
    \exists h_1,h_2.\ h=h_1 \circ h_2 \mbox{ and } s,h_1 \models F_1 \mbox{ and } s,h_2 \models F_2 \\
s,h \models \exists \vec{z}.\ \Pi : F & \Leftrightarrow& \exists\vec{m} \in \nat^{|\vec{z}|}.\ s[\vec{z} \mapsto \vec{m}],h \models \Pi \mbox{ and } s[\vec{z} \mapsto \vec{m}],h \models F
\end{array}\]%
\end{definition}

\section{Satisfiability}   
\label{s-SAT}\label{sec:satisfiability}
 Here we establish upper and lower complexity for the satisfiability
 problem in both the minimal and full variants of our $\SL$ pointer
 arithmetic.

\begin{definition}
\label{defn:gamma}  
 Let\/ $A$ be a symbolic heap of the form\ \
  $$\Pi_A : \textstyle\bigsepstar^\ell_{i=1} t_i\mapsto t_i'$$%
 We describe the heap models\ \mbox{$(s,h)$}\  of\/~$A$ by means of the
 following Presburger
 formula~$\gamma_A$ obtained
 by enriching the pure part\/~$\Pi_A$
 with the constraints on that\/ $t_i$, the allocated
 addresses, must be distinct
 (here $x_1$,..,$x_n$ is the list of all variables):
\begin{equation}%
 \gamma_A({x}_1,..,{x}_n)
 \defeq\ \Pi_A
\ \wedge
  \bigwedge_{1 \leq i<j \leq \ell} ((t_i\leq t_j-1) \vee (t_j\leq t_i-1))\ .
     \label{eq-gamma} 
\end{equation}%
 The above $\gamma_A$ can be easily
 rewritten as a Boolean combination
 of elementary formulas of the form 
 \ \ \mbox{$(x'\leq x+k)$},\ 
 where the `offset' $k$ is a variable or an integer.
\end{definition}

\begin{lemma}\label{l-gamma}
 Any model\/ \mbox{$(s,h)$} for $A$ can be transformed into
 a model for $\gamma_A$, and vice versa.
\end{lemma}
\begin{proof}
 By definition,
 given an\ \mbox{$(s,h)$},\ a model for~$A$, %
 we have\ \mbox{$\Pi_A(s(x_1),..,s(x_n))$}\ is true,
 and $h$ is the disjoint collection of the corresponding cells:
\begin{equation}%
  h =  \bigsepstar^\ell_{i=1} s(t_i)\mapsto s(t_i')
    \label{eq-heap-A}
\end{equation}%
 which implies that\
 \mbox{$\bigwedge_{1 \leq i<j \leq \ell}\ (s(t_i)\neq s(t_j))$}\ .

 Conversely, assume a mapping $s$ provides
 an evaluation \mbox{$(s(x_1),..,s(x_n))$} which makes $\gamma_A$
 true. Then\ \mbox{$\Pi_A(s(x_1),..,s(x_n))$}\ is true,
 and, in addition, we can take a heap $h_A$
 as the disjoint collection of the cells
 in accordance with~(\ref{eq-heap-A}),
 which provides: \mbox{$(s,h_A)\models A$}.
\end{proof}  

\begin{corollary}\label{c-SAT-NP-full}
 Satisfiability is in $\NP$.
\end{corollary}
\begin{proof}
 Follows from Lemma~\ref{l-gamma} and the fact
 that satisfiability for quantifier-free Presburger arithmetic
 belongs to~$\NP$~\cite{Scarpellini:84}.
\end{proof}

\noindent
 Satisfiability is shown $\NP$-hard by reduction from the
\emph{\mbox{$3$-}colourability problem~\cite{Garey-Johnson:79}}.

\begin{problem}[$3$-colourability]
Let \mbox{$G=(V,E)$} be an undirected graph with $n$~vertices
 $v_1,\dots,v_n$.
 The {\em \mbox{$3$-}colourability problem} is to decide
 if there is a \mbox{$3$-}colouring of its vertices
 such that no two adjacent vertices share the same colour.
\end{problem}

\begin{definition}  
\label{defn:3colour_to_sat}  
 Let \mbox{$G=(V,E)$} be an instance graph with $n$~vertices.
 We encode the perfect \mbox{$3$-}colourings of~$G$ with
 the following symbolic heap $A_{G}$.

 We use $c_{i}$ to denote one of the colours, $1$, $2$, or~$3$,
 the vertex~$v_i$ is marked by.

 To encode the fact that no two adjacent vertices $v_i$ and $v_j$ share
 the same colour, we use $c_{i}$ and $c_{j}$ 
 as the addresses, relative to the base-offset $e_{ij}$,
 for two disjoint cells.
 To ensure that all cells allocated
 in question are disjoint,
 with\ \mbox{$1\leq i,j\leq n$},\
 we introduce the numbers $e_{ij}$ as:
\begin{equation}
     e_{ij}  =  i\cdot n^2 + j\cdot n    \label{eq-e-ij}
\end{equation}%
 Our choice is motivated, in particular, by needs of
 Definition~\ref{defn:2-3colour_to_entail} where its $B''_G$
 is guaranteed to be satisfiable whenever
 we allow memory chunks of length~$n$ to accommodate any
 of $n$ distinct colours used in the trivially realizable
 \mbox{$n$-}colouring problem.

\begin{proposition}\label{p-e-ij}  
 Let pairs\ \mbox{$(i,j)$}\ and \ \mbox{$(i',j')$}\ be distinct.
 Then \ \mbox{$|e_{i',j'} - e_{ij}| \geq n$}\ \
\end{proposition}

\noindent Formally, we define $A_G$ to be the following
 quantifier-free symbolic heap:
\comment{
$$\begin{array}{@{}l}
 \bigwedge_{i=1}^{n}(c_0+1\leq{c_{i}}\leq c_0+3)
\wedge
 \bigwedge_{(v_i,v_j)\in E}(c_0+1\leq \widetilde{c_{ij}}\leq c_0+3)
\colon 
\\ \bigsepstar_{(v_i,v_j)\in E}\,
 c_{i} + e_{ij}\mapsto \nil{}\ *\ c_{j} + e_{ij} \mapsto \nil{}\ *\
 \widetilde{c_{ij}} + e_{ij} \mapsto \nil{}
\end{array}$$%
} 
\begin{equation}%
 \bigwedge_{i=1}^{n}(c_0+1\leq{c_{i}}\leq c_0+3)\colon
 \bigsepstar_{(v_i,v_j)\in E}\,
 c_{i} + e_{ij}\mapsto \nil{}\ *\ c_{j} + e_{ij} \mapsto \nil{}
                        \label{eq-sat-NP-hard}
\end{equation}%
Notice that $A_G$ is in
 minimal pointer arithmetic{}. 
\end{definition}


\begin{lemma}\label{lem:3color_to_sat}
 Let $G$ be an instance of the \mbox{$3$-}colouring problem.
 Then $A_G$ from Definition~\ref{defn:3colour_to_sat}
 is satisfiable iff there is a perfect \mbox{$3$-}colouring of\/~$G$.
\end{lemma}   
\begin{proof}
 Any perfect \mbox{$3$-}colouring of\/~$G$, with vertices~$v_i$
 labelled by colours~$b_i$, yields a model \mbox{$(s,h)$}
 for $A_G$ with a stack~$s$ defined as\ \mbox{$s(c_i)=s(c_0)+b_i$}.\
 The corresponding cells,\ \
 \mbox{$s(c_{i}) + e_{ij}\mapsto \nil{}$},\ \
 are all disjoint because of Proposition~\ref{p-e-ij}.  

 Conversely, given a model \mbox{$(s,h)$} for $A_G$, we label each
 of the vertices~$v_i$ by the colour\ \mbox{$b_i = s(c_i)-s(c_0)$},
 providing a perfect \mbox{$3$-}colouring of\/~$G$.
\end{proof}  


\begin{theorem}\label{t-SAT-NP-1}  
 Satisfiability is\ \mbox{$\NP$-}hard,
 even for quantifier-free symbolic heaps $A$
 in minimal pointer arithmetic.
\end{theorem}
\begin{proof}
 From Lemma~\ref{lem:3color_to_sat}. 
\end{proof}

\begin{corollary}\label{c-SAT-NP}
 Satisfiability is\ \mbox{$\NP$-}complete,
 even for quantifier-free symbolic heaps $A$
 in minimal pointer arithmetic.
\end{corollary}

\subsection{About the small model property} 
\label{s-SAT-small-model}


 As for the size of models for symbolic heaps
 in Corollary~\ref{c-SAT-NP-full}, we establish
 the following small model property
 (that is~\cite{Antonopoulos-etal:14},
 any satisfiable formula~$A$ has a model of size polynomial
 in the size of\/~$A$)
 but not for full pointer arithmetic,
 cf.~Remark~\ref{r-small-model}.

\begin{remark}\label{r-small-model}  
 On the contrary, no small model property is valid
 whenever we allow \mbox{$x\leq x'+k$},\
 with $k$ being a variable. 

\noindent
 Let $A_n$ be a symbolic heap of the form (here \mbox{$k_0=0$})
 $$ A_n \defeq\ 
  \bigwedge_{i=0}^{n-1} (x_{i+1} = c_0+k_{i+1}> x_{i}+k_{i}) \colon
  \bigsepstar_{i=1}^{n} x_{i}\mapsto \nil{}
 $$%
 Then we have that\
 \mbox{$\bigwedge_{i=0}^{n-1} (s(k_{i+1}) > 2s(k_{i}))$}\
 for any model \mbox{$(s,h)$} of $A_n$, which implies
\ \mbox{$\bigwedge_{i=0}^{n-1} (s(x_{i+1}) > 2^{i+1})$}.\
 Thus, all models of~$A_n$ necessarily require
 (the distances between) at least a half of addresses in~$h$
  to be of exponential size.
\qed
\end{remark}

\comment{
\begin{equation}%
 \gamma_A({x}_1,..,{x}_n)
 \defeq\ \Pi_A
\ \wedge
  \bigwedge_{1 \leq i<j \leq \ell} ((t_i\leq t_j-1) \vee (t_j\leq t_i-1))\ .
                             \label{eq-gamma}
\end{equation}%
} 


 In order to prove the small model property,
 we need a more workable specification of $\gamma_A$:
\begin{definition}
\label{defn:gamma-1}
 Let\/ $A$ be a symbolic heap under constraints
 from Theorem~\ref{t-small}.
 Then we rewrite its $\gamma_A$ (see Definition~\ref{defn:gamma})
 as
\begin{equation}%
 \gamma_A({x}_1,..,{x}_n)
 \equiv f_A(Z_1, Z_2,\dots, Z_m)
                           \label{eq-gamma-f}
\end{equation}%
 where\ \mbox{$f_A(z_1, z_2,.., z_m)$}\
 is a Boolean function, and within~(\ref{eq-gamma-f})
 the Boolean variable~$z_i$ is substituted with
 $Z_i$ of the form\ ``\mbox{$x_i'\leq x_i+k_i$}''
 where $k_i$ is a fixed integer.
\end{definition}

\begin{proposition}\label{p-gamma}
 Any model \mbox{$(s,h)$} for a symbolic heap $A$ can be determined by
 a Boolean vector \ \mbox{$\bar\zeta = \zeta_1, \zeta_2,.., \zeta_m$}\
 such that  \ \mbox{$f_A(\zeta_1,\zeta_2,..,\zeta_m) = \top$}\
 and the following system, $\gamma_{A,\bar\zeta}$,
 has an integer solution:
\begin{equation}%
\left\{\begin{array}{lcl}{}%
      Z_1 & \equiv & \zeta_1,
\\{}%
      Z_2 & \equiv & \zeta_2,
\\{}%
      \dots & \dots & \dots,
\\{}%
      Z_m & \equiv & \zeta_m\ .
\end{array}\right.
                                  \label{eq-gamma-system}
\end{equation}%
\end{proposition}%
\begin{proof}
 Given a model \mbox{$(s,h)$} of\/~$A$, we can evaluate
 each of the\/~$Z_i$,
 and then calculate
 the appropriate\ \mbox{$\bar\zeta = \zeta_1, \zeta_2,.., \zeta_m$}
 by means of the equations in~(\ref{eq-gamma-system}).
\end{proof}  

\def\node#1{\widehat{#1}}

\begin{definition}\label{d-graph}
 In its turn, the system $\gamma_{A,\bar\zeta}$,
 (\ref{eq-gamma-system}), will be encoded by a {\em constraint graph},
 $\widetilde{G}_{A,\bar\zeta}$, constructed as follows.

 With each variable~$x_i$, we will associate the node labelled
 by $\node{x_i}$.

 In the case of \mbox{$Z_i\equiv \zeta_i \equiv\top $},
 we depict the arrow from the node~$\node{x_i}$
 to the node~$\node{x_i'}$ and label it with~$k_i$.

 In the case of \mbox{$Z_i\equiv \zeta_i \equiv \bot $},
 which means that \mbox{``$x_i\leq x'_i-k_i-1$''},\
 we depict the opposite arrow from the node~$\node{x_i'}$
 to the node~$\node{x_i}$ and label it with the number
 \ \mbox{$-k_i-1$}.\

 To provide the connectivity we need, we will add, if necessary,
 a ``maximum node'' $\node{x_0}$, with the constraint
 \ \mbox{``$ x_i \leq x_0$''} for all $x_i$.
 Cf. Figure~\ref{f-SAT-min}. 
\end{definition}

\begin{example}\label{e-d-graph} 
 Let\/ $A$ be a symbolic heap of the form:\
$$ \mbox{$(y\leq x) \colon x\mapsto \nil{}\ *\ y\mapsto \nil{}$},$$%
 with its $\gamma_A$ being of the form:\
 \mbox{$(y\leq x) \wedge  ((x\leq y-1) \vee (y\leq x-1))$}.

\noindent
 Clearly,
 \mbox{$\gamma_A(x,y) \equiv \gamma_1(x,y)\vee \gamma_2(x,y)$},
 where
 $$\gamma_1(x,y) = (y\leq x) \wedge(x\leq y-1),\quad
    \gamma_2(x,y) = (y\leq x) \wedge(y\leq x-1). $$
 In Figure~\ref{f-SAT-min}  
 we show the constraint graphs for $\gamma_1$ and $\gamma_2$, resp.
 Notice that, because of\ \mbox{$y\leq x$}, the node
 $\node{x}$ is a ``maximum node'' in both cases.

 In the case of~(a), we have no solution. Namely,
 there is a negative cycle of the form\
 \mbox{$\node{x}\, \stackrel{0}{\longrightarrow}\,
 \node{y}\, \stackrel{-1}{\longrightarrow}\node{x}\,$},
 which provides a contradictory\ \mbox{$x\leq x-1$}.

 In the case of~(b), the minimal weighted path from
 $\node{x}$ to $\node{y}$ is of the weight~$-1$,
 which guarantees that \mbox{$y=x-1$}
 is a model for $\gamma_A$ and thereby for~$A$.
\end{example}


\begin{figure*}[t]
\begin{center}%
\hspace*{\fill}
     {$\stackrel{{\smallModelYX{24}}} 
  {(a)\ \ \mbox{$\gamma_1 = (y\leq x) \wedge (x\leq y-1)$}}$}  
\hspace*{\fill}
     {$\stackrel{\smallModelXY{24}}
  {(b)\ \ \mbox{$\gamma_2 = (y\leq x) \wedge (y\leq x-1)$}}$}
\hspace*{\fill}
\phantom{xx}
\end{center}
\caption{The small model property:
 The constraint graphs for a symbolic heap $A$
 of the form:\
 \mbox{$(y\leq x) \colon x\mapsto \nil{}\ *\ y\mapsto \nil{}$},
 with its corresponding $\gamma_A$ of the form
 \mbox{$(y\leq x) \wedge ( (x\leq y-1) \vee (y\leq x-1))$}.
}
\label{f-SAT-min}
\end{figure*}

\begin{theorem}[``the small model property'']
\label{c-SAT-small-model}\label{t-small}
 Let\/ $A$ be a satisfiable symbolic heap
 in minimal pointer arithmetic.
 Then we can find a model \mbox{$(s,h)$} for\/~$A$
 in which all values are bounded by\/~$M$,
 which it suffices to take as:\  \mbox{$M=\sum_i (|k_i|+1) $},\
 where $k_i$ ranges over all occurrences of numbers occurred in\/~$A$.
\end{theorem}
\begin{proof}
 According to Proposition~\ref{p-gamma},
 there is a Boolean vector
 \ \mbox{$\bar\zeta = \zeta_1, \zeta_2,.., \zeta_m$}\
 such that 
 the corresponding system, $\gamma_{A,\bar\zeta}$, has a solution.
 Hence, the associated constraint graph,
 $\widetilde{G}_{A,\bar\zeta}$, has no negative cycles,
 see Definition~\ref{d-graph} and Proposition~\ref{p-Cormen}.

 We define our small model with the following mapping~$s$
 with providing an evaluation \mbox{$(s(x_1),..,s(x_n))$}
 which makes $\gamma_A$ true.
 First we define that\ \mbox{$s(x_0)= M$},
 for the ``maximum node''\ $\node{x_0}$
 - so that \ \mbox{$x_i \leq x_0$}\ for all $x_i$.
 Then \mbox{$s(x_i)$} is defined as:\ \mbox{$M+d_i$},\ where
 $d_i$ is the minimal weighted path
 leading from $\node{x_0}$ to $\node{x_i}$.

\noindent
 E.g., 
 in Example~\ref{e-d-graph} 
 the small model
 is given by \mbox{$s(x)= M$}, and\  \mbox{$s(y)= M-1$}.
\end{proof}  

\begin{remark}\label{r-SAT}
 Contrary to Remark~\ref{r-small-model},
 Theorem~\ref{c-SAT-small-model} 
 is valid even for full pointer arithmetic,
 whenever we confine ourselves to the pointer terms
 of the form \mbox{$x+k_0$},
 with $k_0$ being a fixed base-offset,
 but any Boolean combinations
 of the elementary formulas\ \mbox{$(x'=x+k_0)$},\
 \mbox{$(x'\leq x+k_0)$},\ and\ \mbox{$(x'<x+k_0)$},
 are allowed.

 In addition, the corresponding polytime sub-procedures are running
 as the shortest paths procedures with negative weights allowed
 (e.g., Bellman-Ford algorithm),
 with providing polynomials of low degrees.
\comment{
 In a practically important case where
 we allow \mbox{$x'\leq x+k_0$},\ etc.,
 with $k_0$ being a fixed base-offset,
 we can guarantee the small model property
 (that is
 any satisfiable formula~$A$ has a model of size polynomial
 in the size of\/~$A$)
 even for the full pointer arithmetic but with a fixed pointer offset.

 and, on top of that,
 we allow spatial formulas of the arrays the length
 of which is $\leq k_0$
 and lists which length is $\leq k_0$
  where $k_0$ is a fixed integer. See Theorem~\ref{t-P2-upper}.

 we allow any Boolean combinations
 of the elementary formulas\ \mbox{$(x'=x+k_0)$},\
 \mbox{$(x'\leq x+k_0)$},\ and\ \mbox{$(x'<x+k_0)$},
  where $k_0$ is a fixed integer.
} 
\end{remark}

\comment{
We now focus on \emph{entailment} for $\ASL$.  We establish an upper
bound of $\Pi^{\mathrm{EXP}}_1$ in the \emph{weak $\EXP$
  hierarchy}~\cite{Hartmanis85} via an encoding into $\Pi^0_2$ $\pres$,
and a lower bound of $\Pi^P_2$~\cite{Stockmeyer:77}. Moreover,
quantifier-free entailment is $\CoNP$-complete.

\begin{ourproblem}{Entailment for $\ASL$}
 Given symbolic heaps $A,B$, decide if \mbox{$A \models B$}. $A$ may be
 considered quantifier-free; similar to Prop.~\ref{prop:biabd_qf}, the
 existential quantifiers in $B$ may not mention variables appearing in
 the RHS of a $\mapsto$-formula.
\end{ourproblem}
} 

\section{Entailment}
\label{sec:entailment}

 We now focus on {\em the entailment problem}:
\ \mbox{$A\models B$}\ iff
 every model \mbox{$(s,h)$} of\/~$A$ is also a model of\/~$B$.

\begin{definition} \label{d-entail-1}\label{d-entail-1-1}

\noindent
 Let\/ $A$ be a symbolic heap of the form
 $$\mbox{$\Pi_A\colon \bigsepstar^\ell_{i=1} t_i\mapsto t'_i$},$$%
 and\/ $B$ be a symbolic heap of the form\ \
  $$ \mbox{$\exists\bar{y}\,\Pi_B \colon
   \bigsepstar^{\ell'}_{j=1} u_j\mapsto u'_j$},$$%
 both $A$ and\/ $B$ are symbolic heaps
 in the minimal pointer arithmetic.

\noindent
 We express validity of\ \mbox{$A\models B$}, that is,
 every model \mbox{$(s,h)$} of~$A$ is also a model of~$B$,
 by means of the formula
 $\varepsilon_{A,B}$:
\begin{equation}%
  \varepsilon_{A,B}\ =\ 
 \forall \bar x\, (\gamma_A(\bar x) \to
 \exists\bar{y}\,
 (\gamma_B(\bar x, \bar y) \wedge iso(\bar x, \bar y)))
                    \label{eq-entail-0} 
\end{equation}%
 where the following formula,\ \mbox{$ iso(\bar x, \bar y)$},\
 establishes an isomorphism between
 the disjoint collection of the cells:
 \ \ \mbox{$\bigsepstar^\ell_{i=1} t_i\mapsto t'_i$},\ \
 and
 the disjoint collection of the cells:
 \ \ \mbox{$\bigsepstar^{\ell'}_{j=1} u_j\mapsto u'_j$},\
\begin{equation}%
 iso(\bar x, \bar y) =
  \bigwedge_i\bigvee_j ((t_i = u_j)\wedge(t'_i = u'_j))
 \wedge
  \bigwedge_j\bigvee_i ((u_j = t_i)\wedge(u'_j = t'_i))
     \label{eq-entail-0-iso} 
\end{equation}%

 Each of the above $\gamma_A$, $\gamma_B$, and $iso$ can be easily
 rewritten as a Boolean combination
 of elementary formulas of the form 
 \ \ \mbox{$(x'\leq x+k)$},\ 
 where the `offset' $k$ is a variable or an integer
 (in the case of minimal pointer arithmetic, $k$ is a fixed integer).

\noindent
 Thus our $\varepsilon_{A,B}$ can be rewritten as:
\begin{equation}%
 \forall x_1 \forall x_2..\forall x_n
 \exists y_1\exists y_2..\exists y_m\,
 Q(x_1,x_2,..,x_n,y_1, y_2,..,y_m)
        \label{eq-Q} 
\end{equation}%
 where $Q$ is a Boolean combination
 of elementary formulas of the form 
 \ \mbox{$(x'\leq x+k)$}. 
\end{definition} 

\begin{lemma}\label{l-epsilon}
 Any model\/ \mbox{$(s,h)$}, which is a counter-model
 for \mbox{$A\models B$}, can be transformed into
 a model for\/\ \mbox{$\neg\varepsilon_{A,B}$},
 and vice versa.
\end{lemma}
\begin{proof}
 Similar to Lemma~\ref{l-gamma}.
\end{proof}  

\comment{
\begin{proof}
 By definition,
 given an\ \mbox{$(s,h)$},\ a model for~$A$, %
 we have\ \mbox{$\Pi_A(s(x_1),..,s(x_n))$}\ is true,
 and $h$ is the disjoint collection of the corresponding cells:
\begin{equation}%
  h =  \bigsepstar^\ell_{i=1} s(t_i)\mapsto s(t_i')
    \label{eq-heap-A}
\end{equation}%
 which implies that\
 \mbox{$\bigwedge_{1 \leq i<j \leq \ell}\ (s(t_i)\neq s(t_j))$}\ .

 Conversely, assume a mapping $s$ provides
 an evaluation \mbox{$(s(x_1),..,s(x_n))$} which makes $\gamma_A$
 true. Then\ \mbox{$\Pi_A(s(x_1),..,s(x_n))$}\ is true,
 and, in addition, we can take a heap $h_A$
 as the disjoint collection of the cells
 in accordance with~(\ref{eq-heap-A}),
 which provides: \mbox{$(s,h_A)\models A$}.
\end{proof}  
} 

\subsection{Upper and Lower Bounds}\label{s-small-entail}

 Here we establish the following   
 upper and lower bounds for the general quantified entailment problem.
 Namely,
\begin{itemize}
\item[(a)]
 For full pointer arithmetic, the entailment problem
 belongs to the class Presburger $\Pi^0_2$,
 by which we denote, with a quantifier-free $Q$,
 the class of formulas in the \emph{Presburger arithmetic}
 of the form
\begin{equation}%
 \forall x_1 \forall x_2..\forall x_n
 \exists y_1\exists y_2..\exists y_m\,
 Q(x_1,x_2,..,x_n,y_1, y_2,..,y_m).
                                         \label{eq-Q-1}
\end{equation}%

\item[(b)]
 For minimal pointer arithmetic, the entailment problem
 is proved to be at least $\Pi^P_2$-complete, where
 $\Pi^P_2$ is the second class
 in the polynomial time hierarchy~\cite{Stockmeyer:77}.
\end{itemize}

 The crucial difference between Presburger $\Pi^0_2$ and\/
 polynomial $\Pi^P_2$ is that for the latter
 {\em all variables should be polynomially bounded}.

\begin{proposition}\label{t-entail-P2-1-Pres}
 The entailment problem \mbox{$A \models B$}
 with quantified $A$ and\/ $B$
 is in Presburger $\Pi^0_2$.
\end{proposition}
\begin{proof}
 According to Lemma~\ref{l-epsilon},
 \mbox{$A\models B$} is valid iff the following holds:
\begin{equation}%
 \forall \bar x\, (\gamma_A (\bar x)
 \to \exists\bar{y}
 (\gamma_B(\bar x, \bar y) \wedge iso(\bar x, \bar y)))
                             \label{eq-entail-P2-Pres}
\end{equation}%
 The latter belongs to Presburger $\Pi^0_2$.
\end{proof}  

\noindent
 The lower bound is the same:
\begin{proposition}\label{p-entail-P2-hard-Pres} 
 Since we have allowed {\em arbitrary Boolean combinations}
 of the elementary formulas\ \mbox{$(t_1=t_2)$},\
 \mbox{$(t_1\leq t_2)$}, and\  \mbox{$(t_1<t_2)$},
 we can simulate the class Presburger $\Pi^0_2$,
 providing Presburger $\Pi^0_2$ hardness, even
 within the pure part of our language.
\end{proposition}

\comment{
 For minimal pointer arithmetic, at least we can prove that
 its entailment problem is in a smaller class
 of the form~(\ref{eq-Q}) in which all $y_j$ are Boolean variables.
} 

\begin{remark}\label{r-upper-P2}
 The crucial difference between Presburger $\Pi^0_2$ and\/
 polynomial $\Pi^P_2$ is that for the latter
 {\em all variables should be polynomially bounded}.
\footnote{
 According to Theorem~\ref{t-P2-upper},
 given $A$ and\/ $B$, symbolic heaps
 in minimal pointer arithmetic,
\mbox{$A\models B$} is valid if and only if
 within the corresponding form~(\ref{eq-Q})
 representing~(\ref{eq-entail-P2-Pres}),
 all $x_i$ are bounded by \ \mbox{$(n+1)\cdot M$}\
 and all $y_j$ by\ \mbox{$(n+m+2)\cdot M$},\
 where $M$ is defined as:\ \mbox{$M=\sum_i (|k_i|+1) $},\
 with $k_i$ ranging over all occurrences of these `offset' numbers
 occurred in\/ $A$ and\/~$B$.
 Here $Q$ is a Boolean combination
 of the elementary formulas\ \mbox{$(x'=x+k_0)$},\
 \mbox{$(x'\leq x+k_0)$},\ and\ \mbox{$(x'<x+k_0)$},
 where the `offset' $k_0$ is a fixed integer.
} 
\end{remark}  

\subsection{Quantified minimal arithmetic: A lower bound}

 To prove \mbox{$\Pi^P_2$-}hardness in the quantified case
 for the minimal pointer arithmetic,
 we use the following constructions.

\begin{ourproblem}{$2$-round $3$-colourability problem}
 Let \mbox{$G=(V,E)$} be an undirected graph with $n$~vertices
 $v_1, \dots, v_k, v_{k+1}, \dots v_n$,
 and let $v_1, v_2, \dots, v_k$ be its {\em leaves}.
 The problem is to decide
 if every \mbox{$3$-}colouring of the leaves can be extended to a
\mbox{$3$-}colouring of the graph,
 such that no two adjacent vertices share the same colour.
\end{ourproblem}

\begin{definition}  
\label{defn:2-3colour_to_entail}  
 Let \mbox{$G=(V,E)$} be an instance graph with $n$~vertices
 and $k$~leaves.
 In addition to the variables~$c_{i}$
 in Definition~\ref{defn:3colour_to_sat},
 to each edge \mbox{$(v_i,v_j)$} we associate $\widetilde{c_{ij}}$,
 representing the colour ``complementary'' to $c_{i}$ and\/~$c_{j}$.

 To encode the fact that no two adjacent vertices $v_i$ and $v_j$
 share the same colour,
 we intend to use $c_{i}$, $c_{j}$, and $\widetilde{c_{ij}}$
 as the addresses, relative to the base-offset $e_{ij}$,
 for three consecutive cells within a memory chunk of length~$3$,
 which forces the corresponding
  colours, related to $c_{i}$, $c_{j}$, and $\widetilde{c_{ij}}$,
  to form a {\em permutation} of \mbox{$(1,2,3)$}.
 In order to provide a sufficient memory to accommodate
 the disjoint cells in question, we take the numbers $e_{ij}$
 as in Definition~\ref{defn:3colour_to_sat} to satisfy
 Proposition~\ref{p-e-ij}.  

\noindent Formally, we define $A''_G$ to be the following
 quantifier-free symbolic heap:
\begin{equation}%
 (b=c_0+3)\ \wedge\
 \bigwedge_{i=1}^{k}(c_0+1\leq{c_{i}}\leq b)
\colon 
 \bigsepstar_{(v_i,v_j)\in E,\ \ell=1,2,3}\,
      c_0 + e_{ij}+\ell\mapsto \nil{}
                        \label{eq-entail-P2-hard-A}
\end{equation}%
 and $B''_G$ to be the following quantified symbolic heap:
\begin{equation}%
\begin{array}{@{}l}
\exists\bar{z}.\ \displaystyle 
 \bigwedge_{i=1}^{n}\,(c_0+1\leq{c_{i}}\leq b)
\wedge
 \bigwedge_{(v_i,v_j)\in E}\,(c_0+1\leq \widetilde{c_{ij}}\leq b)
\colon 
\\[4ex]\hspace*{5ex} \displaystyle 
 \bigsepstar_{(v_i,v_j)\in E}\,
 c_{i} + e_{ij}\mapsto \nil{}\ *\ c_{j} + e_{ij} \mapsto \nil{}\ *\
 \widetilde{c_{ij}} + e_{ij} \mapsto \nil{}
\end{array}%
                        \label{eq-entail-P2-hard-B}
\end{equation}%
 where the existentially quantified variables $\vec{z}$
 are all variables occurring in\/~$B''_G$ that are not mentioned
 explicitly in\/~$A''_G$. 

 Notice that both $A''_G$ and $B''_G$ are satisfiable and in
 minimal pointer arithmetic.

 $B''_G$~is satisfiable because  $B''_G$~does not impose any bounds
 on~$b$, so that we can use, for instance, $n$ distinct colours,
 which suffices to produce a perfect \mbox{$n$-}colouring
 for any\/~$G$ with $n$ vertices.

 Proposition~\ref{p-e-ij} takes care of making the corresponding cells
 disjoint.
\end{definition}

\comment{
\begin{lemma}\label{lem:3color_to_sat}
 Let $G$ be an instance of the \mbox{$3$-}colouring problem.
 Then $A_G$ from Definition~\ref{defn:3colour_to_sat}
 is satisfiable iff there is a perfect \mbox{$3$-}colouring of\/~$G$.
\end{lemma}
} 

\begin{lemma}\label{l-colour-entail-P2}
 Let\/ $G$ be a
 \mbox{$2$-}round \mbox{$3$-}colouring instance.
 The entailment problem
\ \mbox{$A''_G\models B''_G$}\ is valid iff there is a winning
 strategy for the perfect $3$-colouring of $G$,
 where $A''_G$ and $B''_G$ are
 the symbolic heaps given by
 Definition.~\ref{defn:2-3colour_to_entail}.
\end{lemma}

\begin{proof} 
 Suppose that there is a winning strategy such
 that every \mbox{$3$-}colouring of the leaves can be extended to a
 perfect \mbox{$3$-}colouring of the whole~$G$.
 We will prove that \mbox{$A''_G\models B''_G$}.

\noindent
 Let \mbox{$s,h$} be a stack-heap pair satisfying
 \mbox{$s,h \models A''_G$}.

 The spatial part of $A''_G$ yields a decomposition of~$h$ as
 the disjoint collection of the cells
 (we recall that \mbox{$s(e_{ij}) = e_{ij}$}
 and \mbox{$s(\ell) = \ell$}):
\begin{equation}%
 h =  \bigsepstar_{(v_i,v_j)\in E,\ \ell=1,2,3}\,
 s(c_0)+ e_{ij}+\ell\mapsto \nil{}
              \label{eq-h-A}
\end{equation}%
 and
\ \ \mbox{$\bigwedge_{i=1}^{k}(s(c_0)+1\leq{s(c_{i})}\leq(s(c_0)+3)$}.

 Take the \mbox{$3$-}colouring of the leaves obtained by assigning
 the colours \mbox{$b_{i}$} to the leaves $v_1$, $v_2$,\dots, $v_k$
 resp..
 where \mbox{$b_{i} = s(c_i)- s(c_0)$}. According to the winning strategy, we can assign colours, denote
 them by $b_{i}$, \mbox{$i>k$}, to the rest of vertices\
 $v_{k+1}$, \dots, $v_n$, resp.,
 obtaining a \mbox{$3$-}colouring of the whole~$G$ such that
 no adjacent vertices share the same colour. In addition, we mark edges \mbox{$(v_i,v_j)$}
 by $\widetilde{b}_{ij}$ complementary to $b_{i}$ and~$b_{j}$.

 We extend the stack $s$ for quantified variables in~$B''_G$
 so that for all \mbox{$i\leq k$},
$$ s(c_i) = s(c_0) + b_{i},$$%
 and, for each \mbox{$(v_i, v_j)\in E$}, we have 
 \mbox{$ s(\widetilde{c_{ij}}) =  s(c_0) + 6 - b_{i} - b_{j}$}. The fact that no adjacent vertices $v_i$ and\/ $v_j$
 share the same colour means that\/\ \
 $$\mbox{$(s(c_{i}),\, s(c_{j}),\, s(\widetilde{c_{ij}}))$}$$%
 {\em is a permutation of\/}
 \ $$\mbox{$(s(c_0)+1,\, s(c_0)+2,\, s(c_0)+3)$},$$%
  and, as a result, \mbox{$(s,h)$} is also a model for $B''_G$:
\begin{equation}%
 h = \bigsepstar_{(v_i,v_j)\in E}\,
 s(c_{i}) + e_{ij}\mapsto \nil{}\ *\ s(c_{j}) + e_{ij} \mapsto \nil{}\ *\
 s(\widetilde{c_{ij}}) + e_{ij} \mapsto \nil{}
              \label{eq-h-B}
\end{equation}

As for the opposite direction, let \mbox{$A''_G{}\models B''_G{}$}. Since \mbox{$A''_G{}$} is satisfiable, there is a model \mbox{$(s,h)$} for $A''_G$ so that, in particular, $h$~satisfies (\ref{eq-h-A}).

We will construct the required winning strategy in the following way. Assume a \mbox{$3$-}colouring of the leaves be given by assigning
 colours, say $b_{i}$, to the leaves $v_1$, $v_2$,\dots, $v_k$
 respectively. We modify our original~$s$ to a stack~$s'$
 by defining, for each \mbox{$1\leq i\leq k$},
 $$ s'(c_i) = s(c_0) + b_{i}.$$%
 which does not change the heap~$h$, but provides
$$\bigwedge_{i=1}^{k}(s(c_0)+1\leq{s'(c_{i})}\leq(s(c_0)+3).$$%

 It is clear that the modified\ \mbox{$(s',h)$}\
 is still a model for $A''_G$, and, hence,
 a model for $B''_G$. Then for some
 stack\/~$s_B$, which is extension of\/~$s'$
 to the existentially quantified variables in\/~$B$,
 we get \ \mbox{$(s_B,h) \models B''_G$}.

 For each \mbox{$1\leq i\leq k$},
 \ \mbox{$s_B(c_i) = s'(c_i) = s_B(c_0) + b_{i}$},\
 which means that, for \mbox{$1\leq i\leq k$},
 these \mbox{$s_B(c_i)$} represent correctly
 the original \mbox{$3$-}colouring of the leaves.

 By assigning the colours\ \ \mbox{$b_i=s_B(c_i)-s_B(c_0)$}\ \ 
 to the rest of vertices $v_{k+1}$, $v_{k+2}$,
 \dots, $v_n$ resp.
 we obtain a \mbox{$3$-}colouring of the whole~$G$.

 The spatial part of the form~(\ref{eq-h-B})
 provides that\ \mbox{$s_B(c_i)\neq s_B(c_j)$},\
 which results in that
 no adjacent vertices $v_i$ and\/ $v_j$ share
 the same colours $b_i$ and\/ $b_j$,
 providing a perfect \mbox{$3$-}colouring of\/~$G$.
\end{proof}  


\begin{theorem}\label{t-entail-P2-0} 
 The entailment problem \mbox{$A \models B$} is
 \mbox{$\Pi^P_2$-}hard,
 even for quantifier-free \red{satisfiable} formulas $A$ and
 quantified \red{satisfiable} formulas~$B$,
 both in minimal pointer arithmetic{}.
\end{theorem}
\begin{proof}
  Via the\ \mbox{$2$-}round \mbox{$3$-}colourability problem,
  with Lemma~\ref{l-colour-entail-P2}.
\end{proof}  

\comment{
\begin{corollary}\label{c-entail-P2}
 The entailment problem \mbox{$A \models B$} is
 \mbox{$\Pi^P_2$-}complete,
 even for quantifier-free \red{satisfiable} formulas $A$ and
 quantified \red{satisfiable} formulas~$B$,
 both in minimal pointer arithmetic{}.
\end{corollary}
} 

\comment{
\maxcomment{Perhaps, not now

 For the minimal pointer arithmetic, however,
 we can improve the upper bound significantly,
 up to $\Pi^P_2$,
 by constructing an efficient procedure for the entailment problem,
 in which the corresponding polytime sub-procedures are running
 as the shortest paths procedures with negative weights allowed
 (e.g., Bellman-Ford algorithm),
 with providing polynomials of low degrees
}}


\subsection{Quantifier-free Entailment}

\comment{

 Then we rewrite $\gamma_A$ as
\begin{equation}%
 \gamma_A(\widetilde{x}_1,..,\widetilde{x}_n)
 \equiv f_A(Z_1, Z_2,\dots, Z_m)
                           \label{eq-gamma-f}
\end{equation}%
 where\ \mbox{$f_A(z_1, z_2,\dots, z_m)$}\
 is a Boolean function,
 and for each~$i$,\ the expression $Z_i$, we substitute which for
 the Boolean variable~$z_i$ in~(\ref{eq-gamma-f}),
 is of the form:\ \mbox{$(x'_i\leq x_i+k_i)$},
 here $x_i$ and $x'_i$ are pointer variables and $k_i$ is an integer.
} 

\begin{theorem}\label{t-entail-coNP-1}  
 The entailment problem \mbox{$A \models B$}
 with quantifier-free~$B$ is in $\CoNP$.
\end{theorem}
\begin{proof}   
 \mbox{$A\models B$} is not valid iff the following holds:
\begin{equation}%
 \exists \bar x\, (\gamma_A \land \neg
 (\gamma_B
 \wedge iso(\bar x, \bar y)))
                             \label{eq-entail-coNP}
\end{equation}%
 At this point, we can follow our proof for Theorem~\ref{t-SAT-NP-1}
 to show that satisfiability of~(\ref{eq-entail-coNP})
 belongs to~$\NP$.
\end{proof}  

\begin{remark}\label{r-small-model-entail}  
 (Cf.~Remark~\ref{r-small-model})  
 No small model property is valid
 whenever we allow \mbox{$x\leq x'+k$},\
 with $k$ being a variable. 

\noindent
 Let $A_n$ and $B_n$ be symbolic heaps of the form
 (here \mbox{$k_0=1$}), both satisfiable:
 $$ A_n \defeq\
  \bigwedge_{i=0}^{n-1} (x_{i+1} = c_0+k_{i+1} = x_{i}+k_{i}) \colon
  \bigsepstar_{i=1}^{n} x_{i}\mapsto \nil{}
 $$%
 and
 $$ B_n \defeq\ (x_n \leq x_0)
 \colon
  \bigsepstar_{i=1}^{n} x_{i}\mapsto \nil{}
 $$%
 \mbox{$A_n\models B_n$} is not valid, but
 for any polynomial~$p$, there is a number~$n_0$ such that
 for all \mbox{$n\geq n_0$},
 there is no counter-model of size \mbox{$\leq p(n)$}.
\qed
\end{remark}

\comment{
\maxcomment{
  a POSITIVE aspect, small model, low degree polynomials

    Modify for entailment
}
} 

\begin{theorem}[``the small model property'']
\label{c-SAT-small-model-entail}\label{t-small-entail}
 Given $A$ and\/ $B$, quantifier-free symbolic heaps
 in minimal pointer arithmetic,
 suppose that \mbox{$A\models B$} is not valid.
 Then we can find a counter-model \mbox{$(s,h)$}
 such that\ \mbox{$(s,h)\models A$}\
 but\ \mbox{$(s,h)\not\models B$},\
 in which all values are bounded by\/~$M$,
 which suffices to take as:\  \mbox{$M=\sum_i (|k_i|+1) $},\
 where $k_i$ ranges over all occurrences of numbers
 occurred in\/ $A$ and\/~$B$.
\end{theorem}
\begin{proof}
 Follow the proof of Theorem~\ref{c-SAT-small-model}. 
\end{proof}  

 As for \mbox{$\CoNP$-}hardness
 even for minimal pointer arithmetic, we will use a construction
 similar to Definition~\ref{defn:3colour_to_sat}.

\begin{definition}  
\label{defn:3colour_to_coNP}
 Taking notations from Definition~\ref{defn:3colour_to_sat},
 we introduce a satisfiable $A'_G$ of the form:
\begin{equation}%
 \bigwedge_{i=1}^{n}(c_0+1\leq{c_{i}}\leq b)\colon
 \bigsepstar_{(v_i,v_j)\in E}\,
 c_{i} + e_{ij}\mapsto \nil{}\ *\ c_{j} + e_{ij} \mapsto \nil{}
                        \label{eq-sat-NP-hard-A}
\end{equation}%
 and a satisfiable $B'_G$ of the form:
\begin{equation}%
 (b\geq c_0+4)\ \wedge\ \bigwedge_{i=1}^{n}\,
(c_0+1\leq{c_{i}}\leq b)\colon
\hspace*{-2ex}
 \bigsepstar_{(v_i,v_j)\in E}
\hspace*{-2ex}\,
 c_{i} + e_{ij}\mapsto \nil{}\ *\ c_{j} + e_{ij} \mapsto \nil{}
                        \label{eq-sat-NP-hard-B}
\end{equation}%
\end{definition}

\begin{lemma}\label{lem:3color_to_coNP}
 Let $G$ be an instance of the \mbox{$3$-}colouring problem.
 Then\ \mbox{$A'_G\models B'_G$}
 is not valid iff
 there is a perfect \mbox{$3$-}colouring of\/~$G$.
\end{lemma}   
\begin{proof}
 Any perfect \mbox{$3$-}colouring of\/~$G$ yields
 a model \mbox{$(s,h)$}
 for $A'_G$ with \mbox{$s(b)=s(c_0)+3$},
 which implies that \mbox{$(s,h)\not\models B'_G$}
 because of \mbox{$s(b)\geq s(c_0)+4$} required there.

 Conversely, the implication of the fact that,
 for some model \mbox{$(s,h)$},
 we have \mbox{$(s,h)\models A'_G$}
 and \mbox{$(s,h)\not\models B'_G$}
 is that \mbox{$s(b)\geq s(c_0)+4$} is false.
 With the additional \mbox{$s(b)\leq s(c_0)+3$},\
 \mbox{$(s,h)\models A'_G$} provides
 a perfect \mbox{$3$-}colouring of\/~$G$.
\end{proof}  

\begin{theorem}\label{t-entail-coNP-0} 
 The entailment problem \mbox{$A \models B$} is \mbox{$\CoNP$-}hard,
 even for quantifier-free \red{satisfiable} formulas $A$ and $B$,
 both in minimal pointer arithmetic{}.
\end{theorem}

\begin{corollary}\label{c-entail-coNP-0} 
 The entailment problem \mbox{$A \models B$} is
 \mbox{$\CoNP$-}complete,
 even for the quantifier-free \red{satisfiable} formulas $A$ and $B$,
 both in minimal pointer arithmetic{}.
\end{corollary}


\section{Quantified entailments:
 The $\Pi^P_2$ upper bound}\label{s-P2-upper} 

 The $\Pi^P_2$ lower bound is given in Theorem~\ref{t-entail-P2-0}.
 For the case of quantified entailments in minimal pointer arithmetic,
 we establish here, Theorem~\ref{t-P2-upper},
 an upper bound also of $\Pi^P_2$,
 as well as the small model property.

 In fact we prove that the upper bound is the same,
 so that minimal pointer arithmetic is $\Pi^P_2$-complete,
 even for the full pointer arithmetic but with a fixed pointer offset,
 where we allow any Boolean combinations
 of the elementary formulas\ \mbox{$(x'=x+k_0)$},\
 \mbox{$(x'\leq x+k_0)$},\ and\ \mbox{$(x'<x+k_0)$},
 and, in addition to the points-to formulas,
 we allow spatial formulas of the arrays the length
 of which is $\leq k_0$
 and lists which length is $\leq k_0$
 where $k_0$ is a fixed integer.

\subsection{Entailment: 
 A running example}

\begin{example}\label{e-small-entail} 
 With this example, we illustrate the crucial steps on the road
 to a smaller model.

\noindent
 Assuming, for simplicity,\ \mbox{$x_1< x_2< x_3< x_4$},\
 let $A$  be of the form
{\small
\begin{equation}%
 (x_1 < x_2) \wedge 
 (x_2 < x_3) \wedge 
 (x_3 < x_4) \wedge 
 (x_3 \leq x_2+3)
\colon
 x_1 \mapsto \nil{}\ *\  x_3 \mapsto \nil{} 
                \label{ex-A}
\end{equation}%
}
 and $B$ be of the form
{\small
\begin{equation}%
 \exists y_1\exists y_2\, \exists y_3\exists y_4\,
 (y_2 = x_2) \wedge (y_4 = x_4) \wedge
 (y_2 \leq y_4-5) \wedge 
 (y_3 = y_1+7)
\colon
 y_1 \mapsto \nil{}\ *\  y_3 \mapsto \nil{} 
                \label{ex-B}
\end{equation}%
}%
 Then 
 $\gamma_A$ in fact is a conjunction
\begin{equation}%
 \gamma_A (x_1, x_2, x_3, x_4) =
\left\{\begin{array}{lcl}{}%
      x_1 & \leq & x_2-1,
\\{}%
      x_2 & \leq & x_3-1,
\\{}%
      x_3 & \leq & x_4-1,
\\{}%
      x_3 & \leq & x_2 + 3\ .
\end{array}\right.
                                  \label{ex-system-A}
\end{equation}
 and by Definition~\ref{d-graph},
 we can also construct the corresponding constraint graph,
  $\widetilde{G}_{A}$,
 the labelled edges of which are given as follows:
\begin{equation}%
 \widetilde{G}_{A} =
\left\{\begin{array}{lcl}{}%
   \node{x_1} &\ \stackrel{-1}{\longleftarrow}\ & \node{x_2}
\\{}%
   \node{x_2} &\ \stackrel{-1}{\longleftarrow}\ & \node{x_3}
\\{}%
   \node{x_3} &\ \stackrel{-1}{\longleftarrow}\ & \node{x_4}
\\{}%
   \node{x_3} &\ \stackrel{3}{\longleftarrow}\ & \node{x_2}
\end{array}\right.
                                  \label{ex-system-A-graph}
\end{equation}%

\comment{
\\
\begin{minipage}{.4\textwidth}
\begin{equation}%
 \gamma_A =
\left\{\begin{array}{lcl}{}%
      x_1 & \leq & x_2-1,
\\{}%
      x_2 & \leq & x_3-1,
\\{}%
      x_3 & \leq & x_4-1,
\\{}%
      x_3 & \leq & x_2 + 3\ .
\end{array}\right.
                                  \label{ex-system-A}
\end{equation}
\end{minipage}
\hspace*{\fill}
\begin{minipage}{.4\textwidth}
\begin{equation}%
 \widetilde{G}_{A} =
\left\{\begin{array}{lcl}{}%
   \node{x_1} &\ \stackrel{-1}{\longleftarrow}\ & \node{x_2}
\\{}%
   \node{x_2} &\ \stackrel{-1}{\longleftarrow}\ & \node{x_3}
\\{}%
   \node{x_3} &\ \stackrel{-1}{\longleftarrow}\ & \node{x_4}
\\{}%
   \node{x_3} &\ \stackrel{3}{\longleftarrow}\ & \node{x_2}
\end{array}\right.
                                  \label{ex-system-A-graph}
\end{equation}%
\end{minipage}
\\
} 

\noindent
 Because of an isomorphism between the spacial parts,
 \mbox{$iso(\bar x,\bar y)$},
 here we get the following:
\begin{equation}%
    iso(\bar x,\bar y)\ \equiv \
 ((y_1=x_1) \wedge (y_3=x_3))\bigvee ((y_1=x_3) \wedge (y_3=x_1))
           \label{eq-iso-exa}
\end{equation}%
 so that 
 the corresponding conclusion in (\ref{eq-entail-0}),
\ \mbox{$\exists\bar{y}\,
(\gamma_B(\bar x, \bar y) \wedge iso(\bar x, \bar y))$},\
 can be rewritten as disjunction of the form
\ \ \mbox{$\exists\bar{y}\,
  G_B^1(\bar x, \bar y)\vee G_B^2(\bar x, \bar y)$}\ \
 where $G_B^1$ and $G_B^2$ are given below: 
\vspace*{-2ex}
\begin{center}
\begin{minipage}{.45\textwidth}
\begin{equation}%
 G_B^1 =
\left\{\begin{array}{lcl}%
 \multicolumn{3}{c}{(y_1=x_1) \wedge (y_3=x_3);}
\\
      y_2 & = & x_2,
\\
      y_4 & = & x_4,
\\{}%
      y_2 & \leq & y_4 - 5,
\\{}%
      y_3 & = & y_1 + 7,
\end{array}\right.
                                  \label{ex-system-B-x1-y1}
\end{equation}%
\end{minipage}
\hspace*{\fill}
\begin{minipage}{.45\textwidth}
\begin{equation}%
 G_B^2 =
\left\{\begin{array}{lcl}
 \multicolumn{3}{c}{(y_1=x_3) \wedge (y_3=x_1);}
\\
      y_2 & = & x_2,
\\
      y_4 & = & x_4,
\\{}%
      y_2 & \leq & y_4 - 5,
\\{}%
      y_3 & = & y_1 + 7,
\end{array}\right.
                                  \label{ex-system-B-x1-y3}
\end{equation}%
\end{minipage}
\end{center}

\def\node#1{\widehat{#1}}

\def\CircNode#1
{\,\raisebox{3pt}{\mbox{\circle{16}}
 \makebox(0,0)[c]{$\widehat{#1}$}}\,}

 To simplify the case, notice that,
 for a fixed
 \ \mbox{$\gamma_A (x_1, x_2, x_3, x_4)$}\ from~(\ref{ex-system-A}),
 the right-hand system~(\ref{ex-system-B-x1-y3}) has no solutions
 because of the cycle with the negative weight:
 \ \mbox{$0-7+0-1-1$},\  
 (see Proposition~\ref{p-Cormen}):
 $$
  \CircNode{x_1}~%
  ~\stackrel{0}{\longrightarrow}~\node{y_3}~%
  ~\stackrel{-7}{\longrightarrow}~\node{y_1}~%
  ~\stackrel{0}{\longrightarrow}~~~\CircNode{x_3}~%
  ~\stackrel{-1}{\longrightarrow}~~~\CircNode{x_2}~%
  ~\stackrel{-1}{\longrightarrow}~~~\CircNode{x_1}
$$%
 Therefore we can confine our attention to the left-hand
 system~(\ref{ex-system-B-x1-y1}), so that
 here validity of\ \mbox{$A\models B$}\ is expressed
 by means of the formula
 $\varepsilon_{A,B}$:
\begin{equation}%
  \varepsilon_{A,B}\ \equiv\
 \forall \bar x\, (\gamma_A(\bar x) \to
 \exists\bar{y}\,  G_B^1(\bar x, \bar y))
                             \label{eq-entail-0-1-exa}
\end{equation}
\end{example}  

\def\node#1{\widehat{#1}}

\def\CircNode#1
{\,\raisebox{3pt}{\mbox{\circle{16}}
 \makebox(0,0)[c]{$\widehat{#1}$}}\,}

\subsection*{Example~\ref{e-small-entail}:
 A large counter-model}

\noindent
 In sequel we will show how to find a small counter-model
 for \ \mbox{$A\models B$},\ 
 given the following `large' counter-model \mbox{$(s,h)$}
 defined by the following\/~$s$
 (here $D$ is a very large number, say~$2^{10}$):  
\begin{equation}
\left\{\begin{array}{lcl}%
        s(x_2) & = & s(x_1) + 2D,
\\{}%
        s(x_3) & = & s(x_2) + 2,
\\{}%
        s(x_4) & = & s(x_3) + D,
\end{array}\right.
                          \label{eq-huge}
\end{equation}%
 First our \mbox{$(s,h)$}, a model for~$A$, is determined uniquely by
 the system:
\begin{equation}%
 \gamma_{A,s} (x_1, x_2, x_3, x_4) =
\left\{\begin{array}{lcl}%
        x_2 & = & x_1 + 2D,
\\{}%
        x_3 & = & x_2 + 2,
\\{}%
        x_4 & = & x_3 + D,
\end{array}\right.
                                         \label{eq-huge-0}
\end{equation}

 We treat\ \mbox{$x'=x+k$}\ as a pair of\ \mbox{$x'\leq x+k$}\
 represented by the edge
 \ \mbox{$\node{x}\, \stackrel{k}{\longrightarrow}\, \node{x'}$},\
 and\ \ \mbox{$x\leq x'-k$}\ \
 represented by the edge
 \ \mbox{$\node{x'}\, \stackrel{-k}{\longrightarrow}\, \node{x}$},\
 so that the corresponding constraint graph, $\widetilde{G}_{A,s}$,
 consists of the following pairs of edges
\begin{equation}%
   \node{x_1}\, \stackrel{2D}{\longrightarrow}\, \node{x_2},\ \
   \node{x_2}\, \stackrel{-2D}{\longrightarrow}\, \node{x_1},\ \
   \node{x_2}\, \stackrel{2}{\longrightarrow}\, \node{x_3},\ \
   \node{x_3}\, \stackrel{-2}{\longrightarrow}\, \node{x_2},\ \
   \node{x_3}\, \stackrel{D}{\longrightarrow}\, \node{x_4},\ \
   \node{x_4}\, \stackrel{-D}{\longrightarrow}\, \node{x_3},\
                                       \label{eq-huge-exact}
\end{equation}%

 Secondly, according to~(\ref{eq-entail-0-1-exa}),
 our \mbox{$(s,h)$} is not a model for~$B$ since
 for a fixed\ $\gamma_{A,s}$\ from~(\ref{eq-huge-0}),
 the following system has no solution:
\begin{equation}%
  \gamma_{A,s} \wedge G_B^1 =
\left\{\begin{array}{lcl}%
        x_2 & = & x_1 + 2D,
\\{}%
        x_3 & = & x_2 + 2,
\\{}%
        x_4 & = & x_3 + D,
\\
 \multicolumn{3}{c}{(y_1=x_1) \wedge (y_3=x_3);}
\\
      y_2 & = & x_2,
\\
      y_4 & = & x_4,
\\{}%
      y_2 & \leq & y_4 - 5,
\\{}%
      y_3 & = & y_1 + 7,
\end{array}\right.
                                  \label{ex-system-B-x1-y1+s}
\end{equation}%
 which is the case because
 of the cycle with the negative weight:
 \\ \mbox{$0-5+0-2D+0+7+0+D=-D+2$},\  
 (see Proposition~\ref{p-Cormen}):
\begin{equation}%
  \CircNode{x_4}~%
  ~\stackrel{0}{\longrightarrow}~\node{y_4}~%
  \stackrel{-5}{\longrightarrow}~\node{y_2}~%
  ~\stackrel{0}{\longrightarrow}~~~\CircNode{x_2}~%
  ~\stackrel{-2D}{\longrightarrow}~~~\CircNode{x_1}~%
  ~\stackrel{0}{\longrightarrow}~\node{y_1}~%
  \stackrel{7}{\longrightarrow}~\node{y_3}~%
  ~\stackrel{0}{\longrightarrow}~~~\CircNode{x_3}~%
  ~\stackrel{D}{\longrightarrow}~~~\CircNode{x_4}~%
                                         \label{eq-neg-D}
\end{equation}%


\subsection{Quantified Entailment: An upper bound}
\begin{theorem}\label{t-P2-upper}
 The entailment problem in minimal pointer arithmetic
 belongs to $\Pi^P_2$, which is the second class
 in the polynomial time hierarchy~\cite{Stockmeyer:77}.

 Moreover, given $A$ and\/ $B$, symbolic heaps
 in minimal pointer arithmetic,
\mbox{$A\models B$} is valid if and only if
 within the corresponding formula~(\ref{eq-Q})
 all $x_i$ are bounded by \ \mbox{$(n+1)\cdot M$}\
 and all $y_j$ by\ \mbox{$(n+m+2)\cdot M$},\
 where $M$ is defined as:
   $$M=\sum_i (|k_i|+1) $$%
 with $k_i$ ranging over all occurrences of `offsets' numbers
 occurred in\/ $A$ and\/~$B$.
\end{theorem}
\begin{proof}
 This follows from the small model property
 provided by Theorem~\ref{t-P2-small-model}
\end{proof}  

\begin{remark}\label{r-P2-upper}
 In fact we prove that the upper bound is the same, $\Pi^P_2$.
 so that
 the entailment problem in
 quantified minimal pointer arithmetic is $\Pi^P_2$-complete,
 even for the full pointer arithmetic but with a fixed pointer offset,
 where we allow any Boolean combinations
 of the elementary formulas\ \mbox{$(x'=x+k_0)$},\
 \mbox{$(x'\leq x+k_0)$},\ and\ \mbox{$(x'<x+k_0)$},
 and, on top of that,
 we allow spatial formulas of the arrays the length
 of which is $\leq k_0$
 and lists which length is $\leq k_0$
  where $k_0$ is a fixed integer.
\qed
\end{remark}  

\subsection{Small model property. Quantified Entailment}
 To prove Theorem~\ref{t-P2-upper},
 we rely upon the following small model property for
 quantified minimal pointer arithmetic.

\begin{theorem}[``the small model property'']
\label{t-P2-small-model}
 Given $A$ and\/ $B$, quantified symbolic heaps
 in minimal pointer arithmetic,
 suppose that\ \mbox{$A\models B$}\ is encoded by
 a formula~(\ref{eq-entail-0})
 in Definition~\ref{d-entail-1}.

 In the case where \mbox{$A\models B$} is not valid,
 we can find a counter-model \mbox{$(s,h)$}
 such that\ \mbox{$(s,h)\models A$}\
 but\ \mbox{$(s,h)\not\models B$},\ 
 in which all $x$-values
 are bounded by \ \mbox{$(n+1)\cdot M$}\
 and all $y$-values are bounded by\ \ \mbox{$(n+m+2)\cdot M$},\
 where $M$ is defined as:
    $$M=\sum_i (|k_i|+1) $$%
 with $k_i$ ranging over all occurrences of `offsets' numbers
 occurred in\/ $A$ and\/~$B$.
\end{theorem}


\comment{
\maxcomment{Everywhere - a small degree polynomial,
 in fact, minimal paths with negative weights allowed.

 For the minimal pointer arithmetic, however,
 we can improve the upper bound significantly, 
 up to $\Pi^P_2$, the second class in the polynomial hierarchy,
 by constructing an efficient procedure for the entailment problem,
 in which the corresponding polytime sub-procedures are running
 as the shortest paths procedures with negative weights allowed
 (e.g., Bellman-Ford algorithm),
 with providing polynomials of low degrees
}
} 


\comment{

\begin{proof}
 \mbox{$A\models B$} is not valid iff the following holds:
\begin{equation}%
\exists \bar x\, (\gamma_A(\bar x) \land \forall\bar{y}
 \neg(\gamma_B(\bar x, \bar y)
  \wedge iso(\bar x, \bar y)))
                             \label{eq-entail-P2}
\end{equation}%

} 


\comment{
(\ref{eq-entail-0})
\begin{equation}%
  \varepsilon_{A,B}\ =\ 
 \forall \bar x\, (\gamma_A(\bar x) \to
 \exists\bar{y}\,
 (\gamma_B(\bar x, \bar y) \wedge iso(\bar x, \bar y)))
\end{equation}%
} 


\begin{proof}(Sketch) 

 For the sake of non-negative solutions, with $x_1$ as a ``zero'' node,
 $y_m$  as a ``maximum node'',
 we will assume that
 \ \mbox{$x_1 < x_2 < \dots < x_n$},\ and add, if necessary,
 that\ \mbox{$x_n \leq y_m$},\ and for all $y_j$,
 \ \mbox{$x_1 \leq y_j \leq y_m$}.

 Let\/ \mbox{$(s,h)$} be a concrete counter-model
 for\ \mbox{$A\models B$},\
 such that \ \mbox{$s(x_1)=0$},\ 
 and, as a model for~$A$,
 \ \mbox{$(s,h)$}\ be determined uniquely by
 the system:
\begin{equation}%
 \gamma_{A,s} (x_1, x_2, \dots, x_n) =
 \bigwedge_{i=1}^{n-1} (x_{i+1} = x_{i} +d_{i,i+1})
                    \label{eq-gamma-s}
\end{equation}%
 where  for all\/\ \mbox{$1\leq i<j\leq n$},
 the $d_{ij}$ is defined as: 
\begin{equation}%
      d_{ij} = s(x_{j}) - s(x_{i})   \label{eq-d-ij}
\end{equation}%

 Following Proposition~\ref{p-gamma},
 the fact that\ \mbox{$(s,h)$}\ is not a model for~$B$
 means that for a certain Boolean function $f_{A,B}$,
 whatever a Boolean vector
 \ \mbox{$\bar\zeta = \zeta_1,.., \zeta_\ell$}\
 such that\ \mbox{$f_{A,B}(\zeta_1,..,\zeta_\ell) = \top$}\
 we take,
 the following system, $G_{A,B,s,\bar\zeta}$,
 has no integer solution
 for a fixed\ $\gamma_{A,s}$\ from~(\ref{eq-gamma-s}),
\begin{equation}%
\left\{\begin{array}{lcl}%
\multicolumn{3}{l}{
 \gamma_{A,s} (x_1, x_2, \dots, x_n)}
\\
      Z_1 & \equiv & \zeta_1,
\\{}%
      \dots & \dots & \dots,
\\{}%
      Z_\ell & \equiv & \zeta_\ell\ .
\end{array}\right.
                                  \label{eq-gamma-system-s}
\end{equation}%

\subsection*{The small counter-model}

 Given~$M$, we introduce a small counter-model\/ \mbox{$(s',h')$}
 by contracting large gaps $d_{i,i+1}$ to smaller ones, $M$,
 as follows:
\begin{equation}%
 \gamma_{A,s'} (x_1, x_2, \dots, x_n) =
 \bigwedge_{i=1}^{n-1} (x_{i+1} = x_{i} +d'_{i,i+1})
                    \label{eq-gamma-s'}
\end{equation}%
 where
\begin{equation}%
  s'(x_{i+1}) :=
\left\{\begin{array}{lcl}{}%
    s'(x_{i}) + d_{i,i+1},&& \mbox{if $d_{i,i+1}\leq M$}
\\{}%
    s'(x_{i}) + M, && \mbox{otherwise}
\end{array}\right.
   \label{eq-s'-i}
\end{equation}%
 For all\/\ \mbox{$1\leq i<j\leq n$},
 we define $d'_{ij}$ as: 
\begin{equation}%
      d'_{ij} = s'(x_{j}) - s'(x_{i})   \label{eq-d'-ij}
\end{equation}%

\subsection*{Example~\ref{e-small-entail}:
 On the edge of disaster}

 Thus, within Example~\ref{e-small-entail} 
 a smaller model \mbox{$(s',h')$}
 is defined with the following\/~$s'$
\begin{equation}
\left\{\begin{array}{lcl}%
        s'(x_2) & = & s'(x_1) + M,
\\{}%
        s'(x_3) & = & s'(x_2) + 2,
\\{}%
        s'(x_4) & = & s'(x_3) + M.
\end{array}\right.
                          \label{eq-small}
\end{equation}%

 To show that \mbox{$(s',h')$} is not a model for~$B$,
 we have to prove that
 the following system has no solution,
 cf.~(\ref{ex-system-B-x1-y1+s}):
\begin{equation}%
  \gamma_{A,s'} \wedge G_B^1 =
\left\{\begin{array}{lcl}%
        x_2 & = & x_1 + M,
\\{}%
        x_3 & = & x_2 + 2,
\\{}%
        x_4 & = & x_3 + M,
\\
 \multicolumn{3}{c}{(y_1=x_1) \wedge (y_3=x_3);}
\\
      y_2 & = & x_2,
\\
      y_4 & = & x_4,
\\{}%
      y_2 & \leq & y_4 - 5,
\\{}%
      y_3 & = & y_1 + 7,
\end{array}\right.
                     \label{ex-system-B-x1-y1+s'}
\end{equation}%

 A natural idea behind our construction
 to detect a cycle with the negative weight for \mbox{$(s',h')$},
 is to take (\ref{eq-neg-D}) defined in terms of\ \mbox{$(s,h)$},
 and then transform it into a hopefully negative cycle
 in terms of\ \mbox{$(s',h')$} by
 replacing its large $D$ and $2D$
 with the modest~$M$, resulting in a cycle of the form
\begin{equation}%
  \CircNode{x_4}~%
  ~\stackrel{0}{\longrightarrow}~\node{y_4}~%
  \stackrel{-5}{\longrightarrow}~\node{y_2}~%
  ~\stackrel{0}{\longrightarrow}~~~\CircNode{x_2}~%
  ~\stackrel{-M}{\longrightarrow}~~~\CircNode{x_1}~%
  ~\stackrel{0}{\longrightarrow}~\node{y_1}~%
  \stackrel{7}{\longrightarrow}~\node{y_3}~%
  ~\stackrel{0}{\longrightarrow}~~~\CircNode{x_3}~%
  ~\stackrel{M}{\longrightarrow}~~~\CircNode{x_4}~%
                                         \label{eq-plus-D}
\end{equation}%
 But the weight of this cycle happens to be {\bf positive}.
\qed

 The challenge to our construction can be resolved
 by the following lemma.
\begin{lemma}\label{l-small}
 Having got a cycle ${\cal C}$ with the negative weight
 for (\ref{eq-gamma-system-s}),
 we can extract
 a smaller cycle with the negative weight
 for (\ref{eq-gamma-system-s}), which is good for\ \mbox{$(s',h')$},\
 as well.  
\end{lemma}
\begin{proof}
 We introduce the following reductions for \ \mbox{$i<j$}:
\begin{itemize}
\item[(a)]
 Let\/\hspace*{4ex}
\begin{equation}%
\mbox{$\CircNode{x_j}~%
  ~\stackrel{}{\longrightarrow}~\node{y}~%
  \stackrel{\sigma}{\Longrightarrow}~\node{y'}~%
  \stackrel{}{\longrightarrow}~~\CircNode{x_i}~%
$}
                         \label{eq-a}
\end{equation}%
 be a part of\/~${\cal C}$, 
 which does not use edges from  $\gamma_{A,s}$, see~(\ref{eq-gamma-s}).
 Here $\sigma$ is the sum of all integers the edges invoked in
 this part are labelled by.

 We consider two cases:

 \begin{itemize}
 \item[(a1)] Let\/ \ \mbox{$d_{ij}+\sigma \geq 0$}.\ \

 Then we replace the above part~(\ref{eq-a})
 with 
\begin{equation}%
\mbox{$\CircNode{x_j}~%
  ~\stackrel{-d_{ij}}{\longrightarrow}~~\CircNode{x_i}~%
$}
                         \label{eq-a1}
\end{equation}%

 Since\ \mbox{$-d_{ij}\leq\sigma$},\ the weight of
 the whole updated ${\cal C}$ remains negative.

 E.g.,
 in Example~\ref{e-small-entail} with its negative~(\ref{eq-neg-D}),
 the following part of this cycle:
$$
  \CircNode{x_4}~%
  ~\stackrel{0}{\longrightarrow}~\node{y_4}~%
  \stackrel{-5}{\longrightarrow}~\node{y_2}~%
  ~\stackrel{0}{\longrightarrow}~~~\CircNode{x_2}~%
 $$%
 can be replaced with
$$
  \CircNode{x_4}~%
  ~\stackrel{-D}{\longrightarrow}~~~\CircNode{x_3}~%
  ~\stackrel{-2}{\longrightarrow}~~~\CircNode{x_2}~%
 $$%
 resulting in a still negative cycle
\begin{equation}%
  \CircNode{x_4}~%
  ~\stackrel{-D}{\longrightarrow}~~~\CircNode{x_3}~%
  ~\stackrel{-2}{\longrightarrow}~~~\CircNode{x_2}~%
  ~\stackrel{-2D}{\longrightarrow}~~~\CircNode{x_1}~%
  ~\stackrel{0}{\longrightarrow}~\node{y_1}~%
  \stackrel{7}{\longrightarrow}~\node{y_3}~%
  ~\stackrel{0}{\longrightarrow}~~~\CircNode{x_3}~%
  ~\stackrel{D}{\longrightarrow}~~~\CircNode{x_4}~%
                                         \label{eq-neg-D-1}
\end{equation}%
 \item[(a2)] Let\/ \ \mbox{$d_{ij}+\sigma < 0$}.\ \

  Then we can identify the following cycle with a negative weight:
 \hspace*{4ex}
\begin{equation}%
\mbox{$\CircNode{x_j}~%
  ~\stackrel{}{\longrightarrow}~\node{y}~%
  \stackrel{\sigma}{\Longrightarrow}~\node{y'}~%
  \stackrel{}{\longrightarrow}~~\CircNode{x_i}~%
  \stackrel{d_{ij}}{\longrightarrow}~~\CircNode{x_j}~%
$}
                \label{eq-a2}
\end{equation}

 Since \ \mbox{$d_{ij}< -\sigma \leq M$},\ we have
 \mbox{$d'_{ij}=d_{ij}$}, and hence   
 this smaller cycle with the negative weight
 is good for\ \mbox{$(s',h')$},\ as well.  
 \end{itemize}

\item[(b)]
 Let\/\hspace*{4ex}
\begin{equation}%
\mbox{$\CircNode{x_i}~%
  ~\stackrel{}{\longrightarrow}~\node{y}~%
  \stackrel{\sigma}{\Longrightarrow}~\node{y'}~%
  \stackrel{}{\longrightarrow}~~\CircNode{x_j}~%
$}\hspace*{3ex}%
                                         \label{eq-b}
\end{equation}%
 be a part of ${\cal C}$, 
 which does not use edges from  $\gamma_{A,s}$, see~(\ref{eq-gamma-s}).
 Here $\sigma$ is the sum of all integers the edges invoked in
 this part are labelled by.

 \begin{itemize}
 \item[(b1)] Let\/ \ \mbox{$d_{ij}\leq \sigma $}.\ \

 Then we replace the above part~(\ref{eq-b})
 with 
\begin{equation}%
\mbox{$\CircNode{x_i}~%
  ~\stackrel{d_{ij}}{\longrightarrow}~~\CircNode{x_j}~%
$}
                         \label{eq-b1}
\end{equation}%

 Since\ \mbox{$d_{ij}\leq\sigma$},\ the weight of
 the whole updated ${\cal C}$ remains negative.

 \item[(b2)] Let\/ \ \mbox{$d_{ij}>\sigma $}.\ \

  Then we can identify the following cycle with a negative weight:
 \hspace*{4ex}
\begin{equation}%
\mbox{$\CircNode{x_i}~%
  ~\stackrel{}{\longrightarrow}~\node{y}~%
  \stackrel{\sigma}{\Longrightarrow}~\node{y'}~%
  \stackrel{}{\longrightarrow}~~\CircNode{x_j}~%
  ~\stackrel{\ -d_{ij}}{\longrightarrow}~~\CircNode{x_i}~%
$}\hspace*{3ex}%
                \label{eq-b2}
\end{equation}

 Suppose that for all $k$ such that\ \mbox{$i\leq k< j$},
 \ \mbox{$d_{k,k+1}\leq M$}.
 Then \mbox{$d'_{ij}=d_{ij}$}, and hence   
 this smaller cycle with the negative weight
 is good for\ \mbox{$(s',h')$},\ as well.  

 Otherwise, for some $k$ such that\ \mbox{$i\leq k< j$},
 \ \mbox{$d_{k,k+1} > M$},\ and thereby by construction
 \ \mbox{$d'_{k,k+1} = M$},\
 and, hence, \ \mbox{$d'_{ij} \geq M$}.\

 Then the following cycle defined in terms of\ \mbox{$(s',h')$},\
\begin{equation}%
\mbox{$\CircNode{x_i}~%
  ~\stackrel{}{\longrightarrow}~\node{y}~%
  \stackrel{\sigma}{\Longrightarrow}~\node{y'}~%
  \stackrel{}{\longrightarrow}~~\CircNode{x_j}~%
  ~\stackrel{\ -d'_{ij}}{\longrightarrow}~~\CircNode{x_i}~%
$}\hspace*{3ex}%
                \label{eq-b2'}
\end{equation}
 is of negative weight, since
 \ \ \mbox{$\sigma-d'_{ij}\leq \sigma- M <0 $}.\ \ 

 E.g.,
 in Example~\ref{e-small-entail} with its negative~(\ref{eq-neg-D}),
 the following part of this cycle:
$$
  \CircNode{x_1}~%
  ~\stackrel{0}{\longrightarrow}~\node{y_1}~%
  \stackrel{7}{\longrightarrow}~\node{y_3}~%
  ~\stackrel{0}{\longrightarrow}~~~\CircNode{x_3}~%
 $$%
 provides a shorter negative cycle in terms of\ \mbox{$(s,h)$}:
$$
  \CircNode{x_1}~%
  ~\stackrel{0}{\longrightarrow}~\node{y_1}~%
  \stackrel{7}{\longrightarrow}~\node{y_3}~%
  ~\stackrel{0}{\longrightarrow}~~~\CircNode{x_3}~%
  ~\stackrel{-2}{\longrightarrow}~~~\CircNode{x_2}~%
  ~\stackrel{-2D}{\longrightarrow}~~~\CircNode{x_1}~%
 $$%
 which can be transformed into a negative cycle
 in terms of\ \mbox{$(s',h')$}:
$$
  \CircNode{x_1}~%
  ~\stackrel{0}{\longrightarrow}~\node{y_1}~%
  \stackrel{7}{\longrightarrow}~\node{y_3}~%
  ~\stackrel{0}{\longrightarrow}~~~\CircNode{x_3}~%
  ~\stackrel{-2}{\longrightarrow}~~~\CircNode{x_2}~%
  ~\stackrel{-M}{\longrightarrow}~~~\CircNode{x_1}~%
 $$%
\end{itemize}
\end{itemize}
 {\bf NB:} We can prove that always the case~(a2) or case~(b2)
 must happen.
\end{proof}  

\noindent
 This concludes the proof of Lemma~\ref{l-small} and thereby
 of Theorem~\ref{t-P2-small-model}.
 \end{proof}  

\comment{
\begin{theorem}\label{t-P2-upper}
 The entailment problem in minimal pointer arithmetic
 belongs to $\Pi^P_2$, which is the second class
 in the polynomial time hierarchy~\cite{Stockmeyer:77}.

 Moreover, given $A$ and\/ $B$, symbolic heaps
 in minimal pointer arithmetic,
\mbox{$A\models B$} is valid if and only if
 within the corresponding formula~(\ref{eq-Q})
 all $x_i$ are bounded by \ \mbox{$(n+1)\cdot M$}\
 and all $y_j$ by\ \mbox{$(n+m+2)\cdot M$},\
 where $M$ is defined as:
   $$M=\sum_i (|k_i|+1) $$%
 with $k_i$ ranging over all occurrences of `offsets' numbers
 occurred in\/ $A$ and\/~$B$.
\end{theorem}
} 

\comment{
 For the minimal pointer arithmetic, however,
 we can improve the upper bound significantly, 
 up to $\Pi^P_2$, the second class in the polynomial hierarchy,
 by constructing an efficient procedure for the entailment problem,
 in which the corresponding polytime sub-procedures are running
 as the shortest paths procedures with negative weights allowed
 (e.g., Bellman-Ford algorithm),
 with providing polynomials of low degrees
} 

\begin{remark}\label{r-small-low-degree}
 The proof of Theorem~\ref{t-P2-small-model}
 provides quite efficient procedures for the entailment problem
 in Theorem~\ref{t-P2-upper},
 in which the corresponding polytime sub-procedures are running
 as the shortest paths procedures with negative weights allowed
 with providing polynomials of low degrees.
\qed
\end{remark}


\section{Conclusions}  
\label{sec:conclusion}

In this paper, we study the points-to fragment of symbolic-heap
separation logic extended with pointer arithmetic, both in a minimal
form allowing only conjunctions of difference constraints
 \mbox{$x' \leq x \pm k$}, and in a fuller form
 admitting Boolean combinations of elementary formulas over 
 pointer/offset sums.
 We establish upper and lower complexity bounds for
satisfiability and quantified/unquantified entailment for both our
variants of $\SL$ pointer arithmetic, as summarised in
Table~\ref{t-summary}.

  Perhaps surprisingly, we find that polynomial
time algorithms are out of reach even for minimal $\SL$ pointer
arithmetic: for example, satisfiability is already $\NP$-complete and
quantifier-free entailment is $\CoNP$-complete.  However, moving to full
rather than minimal pointer arithmetic incurs a surprisingly small
complexity cost: only quantified entailments become harder
($\Pi^\mathrm{EXP}_1$ as opposed to $\Pi^P_2$), although the small model
property is lost.

 We point out that, for the case of quantified
entailments in minimal pointer arithmetic,
 we establish here an upper bound also of $\Pi^P_2$,
 as well as the small model property.


We note that some of our upper bound complexity results can be seen as
following already from our earlier results for \emph{array separation
  logic}, where we allow array predicates $\mathrm{array}(x,y)$ as well
as pointers and arithmetic constraints.  Of course, pointer arithmetic
is often an essential feature in reasoning about array-manipulating
programs. The main value of our findings, we believe, is in our
\emph{lower} bound complexity results, which show that $\NP$-hardness or
worse is an inevitable consequence of admitting pointer arithmetic of
almost any kind.

We remark that our lower-bound results do however rely on the presence
of \emph{pointer} arithmetic, as opposed to arithmetic per se.  If
pointers and data values are strictly distinguished and arithmetic
permitted only over data, as is done e.g. in~\cite{Gu-Chen-Wu:16}, then
polynomial-time algorithms may still be achievable in that case.

\bibliographystyle{splncs03}
\bibliography{pointer_arith_refs}


\end{document}

\end{document}